\documentclass[11pt]{article}
\usepackage[T1]{fontenc}
\usepackage[utf8]{inputenc}
\usepackage{lmodern}
\usepackage{amsmath,amssymb,amsthm,mathtools}
\usepackage[letterpaper,margin=1in,headheight=14pt,headsep=7mm]{geometry}
\usepackage{microtype}
\usepackage{needspace}
\usepackage{flafter}
\usepackage{float}
\usepackage[section]{placeins}
\usepackage{xcolor}
\usepackage{graphicx}
\usepackage{tikz}
\usetikzlibrary{arrows.meta,positioning,calc,matrix,fit,decorations.pathreplacing}
\definecolor{aliceblue}{RGB}{32,96,151}
\definecolor{boborange}{RGB}{183,96,27}
\definecolor{answergreen}{RGB}{34,111,81}
\usepackage{booktabs,tabularx,longtable,array}
\usepackage{enumitem}
\usepackage{fancyhdr}
\usepackage{aliascnt}
\usepackage{cite}
\let\small\normalsize
\usepackage[unicode,pdfencoding=auto]{hyperref}
\usepackage[nameinlink,capitalise,noabbrev]{cleveref}
\definecolor{linkblue}{RGB}{32,71,116}
\hypersetup{colorlinks=true,linkcolor=linkblue,citecolor=linkblue,urlcolor=linkblue,
  pdftitle={Lossless Hardness Condensation in Deterministic Communication Complexity},
  pdfauthor={Simon Mackenzie},
  pdfsubject={Hard original submatrices, rank-sensitive counting, and constructive extraction},
  pdfkeywords={communication complexity, hardness condensation, submatrix restriction, matrix rank, protocol covers}}
\setlist{topsep=0.4em,itemsep=0.2em,parsep=0pt}
\AddToHook{cmd/thebibliography/after}{%
  \setlength{\itemsep}{0.2\baselineskip}\setlength{\parsep}{0pt}}
\numberwithin{equation}{section}
\theoremstyle{plain}
\newtheorem{theorem}{Theorem}[section]
\newaliascnt{lemma}{theorem}
\newtheorem{lemma}[lemma]{Lemma}
\aliascntresetthe{lemma}
\newaliascnt{proposition}{theorem}
\newtheorem{proposition}[proposition]{Proposition}
\aliascntresetthe{proposition}
\newaliascnt{corollary}{theorem}
\newtheorem{corollary}[corollary]{Corollary}
\aliascntresetthe{corollary}
\theoremstyle{definition}
\newaliascnt{definition}{theorem}
\newtheorem{definition}[definition]{Definition}
\aliascntresetthe{definition}
\newaliascnt{example}{theorem}

\aliascntresetthe{example}
\theoremstyle{remark}
\newaliascnt{remark}{theorem}

\aliascntresetthe{remark}
\crefname{lemma}{Lemma}{Lemmas}
\crefname{proposition}{Proposition}{Propositions}
\crefname{corollary}{Corollary}{Corollaries}
\crefname{definition}{Definition}{Definitions}
\crefname{example}{Example}{Examples}
\crefname{remark}{Remark}{Remarks}
\newcommand{\D}{\operatorname{D}}

\newcommand{\card}[1]{\lvert#1\rvert}
\newcommand{\ceil}[1]{\left\lceil#1\right\rceil}
\newcommand{\floor}[1]{\left\lfloor#1\right\rfloor}

\newcommand{\R}{\mathbb R}

\DeclareMathOperator{\rank}{rank}
\fancypagestyle{plain}{\fancyhf{}\fancyfoot[C]{\thepage}}
\allowdisplaybreaks[1]

\title{Lossless Hardness Condensation\\
in Deterministic Communication Complexity}
\author{Simon Mackenzie\thanks{\normalsize University of New South Wales (UNSW).\\
Email: \href{mailto:simon.william.mackenzie@gmail.com}{\texttt{simon.william.mackenzie@gmail.com}}.}}
\date{}
\begin{document}
\hypersetup{pageanchor=false}
\begin{titlepage}
\maketitle
\thispagestyle{empty}
\begin{abstract}
A communication problem can have far more possible inputs than its
communication cost would suggest. Must its difficulty already be present
on a much smaller set of inputs? We prove that every finite total Boolean
matrix of deterministic communication complexity $c\ge4$ has a submatrix
on $2^k$ of its original rows and $2^k$ of its original columns, with
$k=\Theta_\varepsilon(c)$ and complexity
at least $(1-\varepsilon)(k+1)$, for every fixed $0<\varepsilon<1$.
Since $k+1$ is the maximum possible cost on such a square, the retained
problem can be arbitrarily close to maximally hard. This answers
affirmatively the lossless condensation question of Hamed Hatami;
G\"o\"os, Newman, Riazanov, and Sokolov (STOC 2024), who recorded it as
Open Problem~2, conjectured a negative answer. Hrube\v{s} previously
guaranteed input length $\Omega(\sqrt c)$.
The same argument gives an original $2^{c-2}$-by-$2^{c-2}$ square
retaining at least $c/3-O(\log c)$ bits of communication complexity.

The proof builds on Hrube\v{s}'s counting and covering argument. We count
submatrices equipped with short communication protocols: a player names
a covering submatrix, then the players run its protocol. This avoids
the loss from converting rectangle partitions into protocols.
A recursion on rectangles makes the argument constructive. For fixed
rational $\varepsilon$ and any target depth $d\ge4$, a deterministic
algorithm returns either a protocol of depth below $d$, or a square
of original inputs at input length $\Theta_\varepsilon(d)$ with the
same near-maximal guarantee. Its running time is $2^{O(2^d)}$ times
a polynomial in the table size. If the original complexity is at least
$d$, the algorithm necessarily returns the square.
An extension gives constant-factor condensation for any fixed number
of number-in-hand players, with bounds independent of the finite output
alphabet.
\end{abstract}

\vfill
\end{titlepage}
\pagenumbering{arabic}
\hypersetup{pageanchor=true}
\section{Introduction}\label{sec:introduction}

Alice and Bob hold separate inputs and want to compute a Boolean
function of them. Communication complexity~\cite{Yao79} asks how many bits they must
exchange in the worst case. The inputs themselves may be much longer
than the best conversation. If every protocol needs $c$ bits, is that
difficulty already visible among $2^{O(c)}$ possible inputs for each
party, or can it depend on a much larger collection of inputs?

A restriction makes this question precise. Write the function as a
matrix $M$, with Alice's inputs indexing rows and Bob's inputs indexing
columns. Choose some original rows and columns and retain their
intersection. The players know the chosen sets in advance; all retained
answers remain unchanged. Restricting inputs can only make the problem
easier. The issue is how much difficulty must survive.

On a $2^k$-by-$2^k$ Boolean matrix, there is always a protocol using
$k+1$ bits: Alice names her row and Bob transmits the answer. Thus
a square of this size is as hard as its input length permits when its
communication cost is close to $k+1$. The condensation question asks
whether every hard matrix contains a square whose input length and
communication cost are both proportional to the original cost.
Following~\cite{GNRS24}, \emph{lossless condensation} means this
constant-factor guarantee, not exact preservation of the original cost.

Hamed Hatami asked this question at the 2022 Banff communication
complexity workshop, and G\"o\"os, Newman, Riazanov, and Sokolov recorded
it as Open Problem~2 of their STOC 2024 paper~\cite{GNRS24}, conjecturing
a negative answer. Hrube\v{s} previously guaranteed input length
$\Omega(\sqrt c)$~\cite[Proposition~5.2]{Hrubes25}. We obtain linear
input length together with communication arbitrarily close to the
maximum possible on the retained square, at the price of a smaller
linear constant.

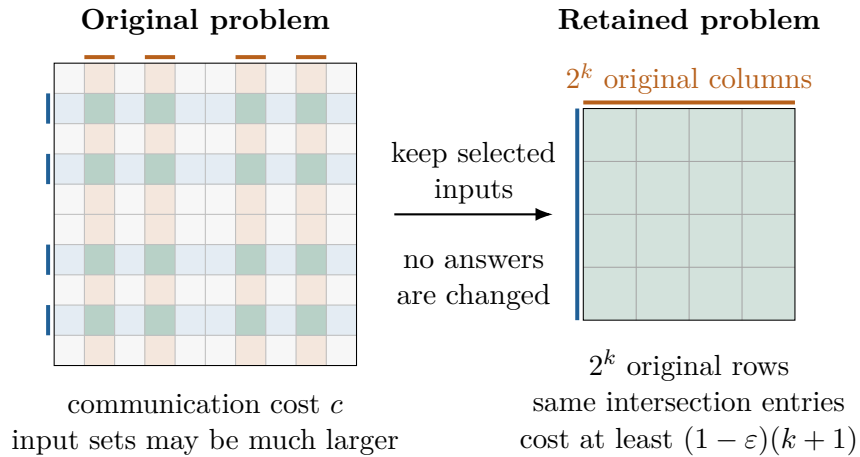
\begin{figure}[H]
\centering
\begin{tikzpicture}[x=1cm,y=1cm,>=Latex,font=\small]
  \node[font=\small\bfseries] at (2,4.55) {Original problem};
  \fill[black!3] (0,0) rectangle (4,4);
  \foreach \i in {1,3,6,8} {
    \fill[aliceblue!12] (0,{.4*\i}) rectangle (4,{.4*\i+.4});
    \draw[aliceblue,line width=1.5pt] (-.08,{.4*\i}) -- (-.08,{.4*\i+.4});
  }
  \foreach \i in {1,3,6,8} {
    \fill[boborange!15] ({.4*\i},0) rectangle ({.4*\i+.4},4);
    \draw[boborange,line width=1.5pt] ({.4*\i},4.08) -- ({.4*\i+.4},4.08);
  }
  \foreach \i in {1,3,6,8} {
    \foreach \j in {1,3,6,8} {
      \fill[answergreen!32] ({.4*\i},{.4*\j}) rectangle ({.4*\i+.4},{.4*\j+.4});
    }
  }
  \foreach \i in {1,...,9} {
    \draw[black!25,thin] ({.4*\i},0) -- ({.4*\i},4);
    \draw[black!25,thin] (0,{.4*\i}) -- (4,{.4*\i});
  }
  \draw (0,0) rectangle (4,4);
  \draw[->,thick] (4.5,2) -- (6.6,2)
    node[midway,above,align=center] {keep selected\\inputs};
  \node[align=center,text width=2.5cm] at (5.55,1.1) {no answers\\are changed};
  \node[font=\small\bfseries] at (8.4,4.55) {Retained problem};
  \fill[answergreen!20] (7,0.6) rectangle (9.8,3.4);
  \foreach \i in {1,2,3} {
    \draw[black!35,thin] ({7+.7*\i},.6) -- ({7+.7*\i},3.4);
    \draw[black!35,thin] (7,{.6+.7*\i}) -- (9.8,{.6+.7*\i});
  }
  \draw (7,.6) rectangle (9.8,3.4);
  \draw[aliceblue,line width=1.5pt] (6.92,.6) -- (6.92,3.4);
  \draw[boborange,line width=1.5pt] (7,3.48) -- (9.8,3.48);
  \node[text=boborange] at (8.4,3.8) {$2^k$ original columns};
  \node[anchor=north,align=center] at (8.4,.35)
    {$2^k$ original rows\\same intersection entries};
  \node[anchor=north,align=center] at (2,-.25)
    {communication cost $c$\\input sets may be much larger};
  \node[align=center] at (8.4,-1)
    {cost at least $(1-\varepsilon)(k+1)$};
\end{tikzpicture}
\caption{What condensation preserves. Select original rows and columns
and keep their intersection.
The theorem guarantees $k=\Theta_\varepsilon(c)$ and nearly maximal
communication on the retained square. Shading marks selected inputs,
not Boolean entry values. The grid is schematic.}
\label{fig:restriction}
\end{figure}

The retained inputs must be few enough to name with $O(c)$ bits,
but numerous enough to support a problem requiring $\Omega(c)$ bits.
\Cref{fig:restriction} shows the object we seek. We first state the
guarantees, then use a small lookup problem to explain the proof.

\subsection{Results}

Throughout the paper, $\D(M)$ denotes
deterministic communication complexity with fixed Boolean outputs at
the leaves of a binary protocol tree. Either player may start. The
protocol determines from the preceding transcript who sends the next
bit; the players need not alternate. \Cref{sec:preliminaries} gives the
full definition.

\begin{theorem}[Lossless condensation]\label{thm:condensation}
Let $M$ be a finite total Boolean matrix with $\D(M)=c\ge4$.
For every $0<\varepsilon<1$, there are an integer $1\le k\le c-1$
and a submatrix $N$ consisting of exactly $2^k$ distinct original rows
and $2^k$ distinct original columns such that
\[
 k\ge\frac{\varepsilon c-4}{3},
 \qquad
 \D(N)\ge(1-\varepsilon)(k+1).
\]
\end{theorem}

For fixed $\varepsilon$, the retained input length is linear in $c$.
Smaller $\varepsilon$ brings the guaranteed complexity closer to the
maximum $k+1$, at the price of a smaller guaranteed $k$. The proof gives
one of two outcomes, according to the rank of $M$: a nonsingular minor
of order $2^k$ with $k\ge(\varepsilon c-4)/3$, whose complexity is
exactly $k+1$; or a square with $k=c-2$ whose complexity exceeds
$(1-\varepsilon)c$, so that all but an $\varepsilon$ fraction of the
original cost survives. There is no rank or distributional assumption.

There is also a direct guarantee in terms of the original cost $c$.
Balancing the two cases, then enlarging the smaller restriction with
original inputs, gives the following consequence.

\begin{corollary}[Retaining original complexity]\label{cor:full-input-length}
Let $M$ be a finite total Boolean matrix with $\D(M)=c\ge4$.
There is an original $2^{c-2}$-by-$2^{c-2}$ submatrix $N$ such that
\[
 \D(N)\ge\frac{c-\ceil{\log(c+1)}-1}{3}.
\]
\end{corollary}

Thus at input length $c-2$ per player, at least one third of the
original communication complexity survives, up to an $O(\log c)$ loss.
This measures retained hardness against $c$; \cref{thm:condensation}
instead lets us approach the maximum possible hardness on a possibly
smaller square. Both are consequences of the same rank-sensitive bound
in \cref{thm:rank-core}.

The proof builds on Hrube\v{s}'s rank-sensitive counting and covering
approach~\cite{Hrubes25}. We count submatrices equipped with short
communication protocols. Covering one player's inputs by such submatrices
lets that player name one and then run its protocol. The added communication
is just the index. Working directly with protocol depth yields the linear
input-length scale; \cref{sec:related-work} compares this with the earlier
bound obtained through rectangle partitions.

For the near-maximal guarantee, \cref{thm:near-maximal} gives a sharper
lower bound on the retained input length $k$. The following algorithm
realizes this guarantee. It
receives the table and a target depth $d$. If it cannot produce a hard
restriction at that scale, it supplies a short protocol for the whole matrix.

\begin{theorem}[Constructive condensation]\label{thm:constructive-main}
Fix a rational $0<\varepsilon<1$. Given a finite total Boolean matrix
$M$ of table size $T$ and an integer $d\ge4$, a deterministic algorithm
returns either an explicit protocol for $M$ of depth less than $d$,
or an original $2^k$-by-$2^k$ submatrix $N$ satisfying
\[
 \begin{gathered}
 1\le k\le d-2,\qquad \D(N)\ge(1-\varepsilon)(k+1),\\
 k\ge\max\left\{\frac{\varepsilon d-4}{3},
     \frac{\varepsilon d-\ceil{\log(d+1)}-3}{2}\right\},
 \end{gathered}
\]
in time $2^{O(2^d)}T^{O(1)}$.
\end{theorem}

No hardness promise is needed to run this algorithm. If $\D(M)\ge d$,
the protocol outcome is impossible, so the algorithm returns the square.
The alternatives need not be disjoint: obtaining the square does not
certify that $\D(M)\ge d$.
Taking $d=\D(M)$ gives a constructive form of \cref{thm:condensation}
and of the sharper bound in \cref{thm:near-maximal} for rational
$\varepsilon$. \Cref{sec:condensation} proves both existence statements
directly for every real $\varepsilon$.

A size bound already permits extraction by exhaustive search when a valid
lower bound $\D(M)\ge d$ is supplied. Combining
this search with a rank test and removal of repeated input values gives
the same guarantee in time $2^{O(2^{(1+\varepsilon/2)d})}T^{O(1)}$
(\cref{sec:exhaustive-search}). Our algorithm improves this running-time
bound's parameter dependence to $2^{O(2^d)}$, while retaining an absolute constant exponent
on the table size. It also makes the counting and covering proof constructive: its
subroutines produce either an explicit short global protocol or hard
original inputs. The running time is doubly exponential in $d$; the
algorithm does not claim polynomial-time condensation.

The main text proves these two-party results. \Cref{app:tradeoffs}
records an alternative tradeoff with faster extraction and input length
within five bits of a supplied lower bound $d$, at a constant fraction
of maximal hardness.
\Cref{sec:multiparty} extends constant-factor condensation to any fixed
number of players in the number-in-hand model: each player knows only their own
input and all communication is public. For $p$ players and original
complexity $c\ge p+2$, a restriction to $2^{O(c)}$ inputs per player
retains complexity at least $\ceil{c/(p+1)}$; \cref{thm:multiparty-main}
gives the exact size bound. These bounds are independent of the finite
output alphabet. This extension is existential; the
algorithmic result above is two-party.

\subsection{Proof ideas}\label{sec:proof-ideas}

The proof works backwards from a short global protocol. If every small
restriction were easy, we would assemble their protocols into a short
protocol for the whole matrix.

One concise quantitative form is \cref{thm:local-global}: for any integer
$s\ge1$, if every nonempty submatrix with at most $2^{4s}$ rows and at
most $2^{4s}$ columns has complexity at most $s$, then the whole matrix
has complexity at most $3s-2$. In other words, uniformly cheap small
restrictions force a short global protocol. The nearly maximal bounds
above use the same argument with a more careful choice of parameters.
We first use a small example to see how local protocols can be combined.

\paragraph{A lookup example.}\label{sec:lookup}
The following matrix will accompany the argument. Alice has two
bits $x=(x_1,x_2)$; Bob has a query $y\in\{0,1,2\}$.
Query zero always asks for the answer zero, and query $j\in\{1,2\}$
asks for $x_j$. Thus
\begin{equation}\label{eq:lookup}
 F(x,0)=0,\qquad F(x,1)=x_1,\qquad F(x,2)=x_2.
\end{equation}
\Cref{fig:lookup} displays every entry.

\begin{figure}[H]
\centering
\begingroup
\renewcommand{\arraystretch}{1.25}
\setlength{\tabcolsep}{11pt}
\begin{tabular}{c|ccc}
 & \textcolor{boborange}{$y=0$}
 & \textcolor{boborange}{$y=1$}
 & \textcolor{boborange}{$y=2$}\\
\hline
\textcolor{aliceblue}{$x=00$}&0&0&0\\
\textcolor{aliceblue}{$x=01$}&0&0&1\\
\textcolor{aliceblue}{$x=10$}&0&1&0\\
\textcolor{aliceblue}{$x=11$}&0&1&1
\end{tabular}
\endgroup
\caption{The lookup matrix $F$. Alice supplies a row; Bob selects
a column. Column zero is constant, and columns one and two read
Alice's first and second bits.}\label{fig:lookup}
\end{figure}

On the column set $C_0=\{0,1\}$, Bob sends a bit distinguishing
the two queries, then Alice sends zero or $x_1$ as appropriate.
On $C_1=\{1,2\}$, Bob likewise indicates his query and Alice sends
the requested bit. Each depth-two protocol works on \emph{all four rows},
and their domains together cover all of Bob's inputs.

Bob can therefore first announce which domain to use: $C_0$ if
$y\in\{0,1\}$, and $C_1$ if $y=2$. Both parties then run the chosen
protocol. The cost is at most three bits: one for the domain index and
two for its protocol. Overlap causes no difficulty because Bob uses a
fixed selection rule. In this small example the cover is no cheaper
than the trivial protocol, in which Alice names her row and Bob sends
the answer; the gain comes when few domains cover many columns.

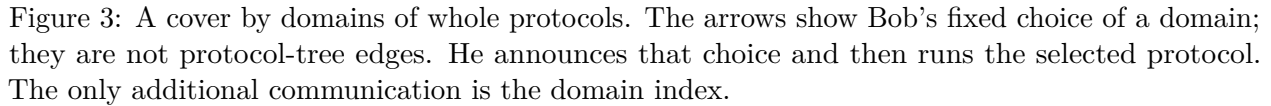
\begin{figure}[H]
\centering
\begin{tikzpicture}[x=1cm,y=1cm,>=Latex,
 every node/.style={font=\small},
 protocolbox/.style={draw,rounded corners=2pt,minimum width=3.2cm,
 minimum height=1.5cm,align=center}]
\node (zero) at (0,1.7) {$y=0$};
\node (one) at (2.6,1.7) {$y=1$};
\node (two) at (5.2,1.7) {$y=2$};
\node[protocolbox,draw=boborange,fill=boborange!5] (left) at (.7,0)
 {$C_0=\{0,1\}$\\Bob: which query?\\Alice: $0$ or $x_1$};
\node[protocolbox,draw=boborange,fill=boborange!5] (right) at (4.4,0)
 {$C_1=\{1,2\}$\\Bob: which query?\\Alice: $x_1$ or $x_2$};
\draw[->] (zero.south) -- ([xshift=-7mm]left.north);
\draw[->] (one.south) -- ([xshift=7mm]left.north);
\draw[->] (two.south) -- ([xshift=7mm]right.north);
\node[align=center] at (2.6,-1.5)
 {Bob chooses a domain using only $y$.\\
 Each protocol accepts every Alice input.};
\end{tikzpicture}
\caption{A cover by domains of whole protocols. The arrows show
Bob's fixed choice of a domain; they are not protocol-tree edges.
He announces that choice and then runs the selected protocol.
The only additional communication is the domain index.}
\label{fig:domain-cover}
\end{figure}

The example isolates the proof's main communication step: a cover by
short protocols gives a global protocol at the cost of naming a domain.
For a general matrix, we must find a cover with few domains. Three
ingredients make this possible, and a rank dichotomy combines them.

\Needspace{9\baselineskip}
\paragraph{Protocol covers.}
Keep every Alice input and restrict only Bob's columns. Fix Alice's
instructions in a short protocol, then include \emph{every} column on
which Bob can complete them correctly. Call this full set a
\emph{protocol domain}. As in the lookup example, a cover
by such domains gives a global protocol: Bob names a domain containing
his column, then the players run its protocol. His choice needs no
information about Alice's input. We count full domains, not every cheap
subset; even a constant matrix can have arbitrarily many cheap subsets.

\paragraph{Rank-sensitive counting.}
Counting Alice's message functions directly would depend on the number
of row labels. Instead, consider the columns accepted at one leaf:
every row reaching that leaf must give the leaf's answer. Discarding
redundant requirements leaves a list whose length is at most the rank
$r\ge1$ of the matrix. Summing the rows on that list encodes the whole
condition in one vector: in a given column, a sum of zero means that
all these rows answer zero, and a sum equal to their number means that
all answer one. The sum is determined by its entries on at most $r$
basis columns, so there are at most $2(r+1)^{r+1}$ leaf conditions.
Intersections and unions combine them along the tree, so partial
protocols of depth $s$ have at most $2^{O(r\log(r+1))\,2^s}$ protocol
domains, however large the matrix is (\cref{lem:domain-count}).

\paragraph{Finite covering.}
Because there are few domains, a finite covering argument gives a
dichotomy: either some set of $O(r\log(r+1)\,2^s)$ columns, about the
logarithm of the number of domains, escapes every protocol of depth
$s$, or one such protocol handles more than half of the uncovered
columns. A matrix of rank $r$ has at most $2^r$ distinct column values,
so repeating the second outcome at most $r$ times gives a cover by at
most $r$ domains, and naming one of them costs $\ceil{\log r}$ bits.
We apply the argument again with rows and columns interchanged.

\paragraph{The rank dichotomy.}
If $\log r$ is a sufficiently large fraction of $c$, a nonsingular
minor gives a square at a linear input-length scale with exactly
maximal complexity (\cref{lem:rank-facts}). This standard minor argument
was already noted in the condensation setting by G\"o\"os et
al.~\cite[Section~1.2]{GNRS24}. The additional step here handles the
complementary low-rank case. Here the two domain
indices are cheap: the total loss is $O(\log r)$ bits. The covering
argument then gives a square with $k=c-2$ retaining more than
$(1-\varepsilon)c$ bits of complexity. \Cref{thm:rank-core,thm:near-maximal}
give the thresholds and constants. Thus rank either supplies the hard
restriction directly or makes the cover inexpensive; no rank promise
is needed.

\paragraph{Making the argument constructive.}
The constructive proof realizes the same objects. Start with rectangles
on which the answer is constant. At an Alice node, join two protocols
on the union of their row sets and intersection of their column sets:
Alice can choose a child from her row alone. At a Bob node, reverse the
roles. This builds executable protocols together with their domains.
The counting bound limits the enumeration, which yields a cover or
hard inputs. We develop the domain and counting arguments in
\cref{sec:domains,sec:counting}, prove condensation in
\cref{sec:condensation}, and give the algorithm in \cref{sec:construction}.
The lookup example accompanies each new ingredient.

\subsection{Related work}\label{sec:related-work}

Hardness condensation, which shrinks the input of a hard problem while
keeping it hard at the new scale, was introduced in circuit complexity
by Buresh-Oppenheim and Santhanam~\cite{BureshOppenheimSanthanam06} and
later used by Razborov in proof complexity~\cite{Razborov16}.
G\"o\"os, Newman, Riazanov, and Sokolov~\cite{GNRS24} study condensation
by restriction, in which the smaller problem keeps original inputs.
They prove lossless condensation for an $n$-variable query function
$f$ lifted with an inner-product gadget on $\Theta(\log^2 n)$ bits,
assuming the deterministic decision-tree complexity of $f$ is
$n^{\Omega(1)}$. They conjecture, by
analogy with their negative result for query complexity, that
deterministic communication complexity does not condense losslessly in
general~\cite[Section~1.2]{GNRS24}; \cref{thm:condensation} refutes
this conjecture. The CCC 2026 paper of Kayal and
coauthors~\cite[Section~6]{KayalEtAl26} still lists the question as open.

Hrube\v{s}'s unconditional bound~\cite[Proposition~5.2]{Hrubes25}
passes through the \emph{1-partition number} $\chi_1(M)$: the least
number of disjoint all-one rectangles partitioning the one-entries.
Yannakakis's bound $\D(M)\le O((\log(\chi_1(M)+1))^2)$~\cite[Lemma~1]{Yannakakis91}
supplies a square-root lower bound on $\log\chi_1(M)$ in terms of
$\D(M)$. Its quadratic dependence is tight up to polylogarithmic
factors~\cite[Theorem~2]{GoosPitassiWatson18}. Hrube\v{s}'s theorem preserves
a constant fraction of this partition-number scale, giving input length
$\Omega(\sqrt{\D(M)})$ and communication proportional to that length.
His proof already combines rank-sensitive counting of maximal input
sets~\cite[Lemma~4.2]{Hrubes25} with a finite covering
argument~\cite[Lemma~3.1]{Hrubes25}.
We extend this count-and-cover approach to domains of whole protocols.
Every covering domain already carries an executable protocol, so the
extra communication only names a member of the cover. Working directly
with protocol depth gives a linear input-length scale.

The deterministic model matters. Hambardzumyan, Hatami,
and Hatami~\cite[Theorem~3.1]{HambardzumyanHatamiHatami22} give matrices with
growing public-coin randomized complexity whose submatrices up to the
square-root side length have constant randomized complexity. As
G\"o\"os, Newman, Riazanov, and Sokolov observe~\cite[Section~1.3]{GNRS24}, these parameters
rule out condensation in that model even with a polynomial loss. The negative
query results of~\cite{GNRS24,KayalEtAl26} concern coordinate restrictions,
whereas here we may choose arbitrary rows and columns.

The leaf-condition count has precedents in graph algorithms, where
unions of neighbourhoods are compressed using matrix
rank~\cite{BuiXuanTelleVatshelle10,OumSaetherVatshelle14}; a leaf's accepted
columns have this form after taking complements. We use the row-sum
encoding of Oum, S{\ae}ther, and Vatshelle to count leaf conditions, then
combine them along the protocol tree. \Cref{sec:counting} gives the
details. This composition yields the nearly maximal retained hardness.

Listing all distinct intersections of a family of sets is a standard
enumeration problem~\cite{MaryStrozecki19}. Our constructive proof uses
such a procedure for leaf domains, with list size bounded by the
row-sum count. Chandran, Issac, and
Karrenbauer~\cite[Section~2]{ChandranIssacKarrenbauer17} use basis guessing
and independent row reconstruction for the related problem of finding
biclique partitions. Here we join monochromatic rectangles along a
protocol tree. The additional requirement is to produce a family of
protocol-bearing domains such that every cheap restriction is contained
in some member, rather than to construct a biclique partition of the
whole matrix.

More broadly, small sets witnessing circuit hardness appear in
Lipton--Young anticheckers~\cite{LiptonYoung94}.
Here communication cost must control the size of an original restriction,
despite the unbounded description length of local message functions.
All counting and reconstruction arguments used here are proved in the
paper.

We now specify the protocol model and begin the elementary proof.

\section{Preliminaries}\label{sec:preliminaries}

We use the standard deterministic protocol-tree framework; see
Rao and Yehudayoff~\cite[Chapter~1]{RaoYehudayoff20}.
The output convention, which matters for our integer bounds, is stated below.

All input sets in the paper are finite and nonempty. A Boolean
communication matrix is a function
\[
 M:X\times Y\longrightarrow\{0,1\}.
\]
Alice knows $x\in X$, Bob knows $y\in Y$, and both know the matrix.
The word \emph{total} means that an answer is specified for every
pair in $X\times Y$.

A deterministic protocol is a finite rooted binary tree. At each
internal node, the tree specifies which player sends the next bit
and a function from that player's input set to $\{0,1\}$. We call
that player the \emph{speaker} at the node. The transmitted bit determines
the next child. A node's identity already specifies the preceding
transcript, so its message function may depend on that transcript.
Each leaf has a fixed output bit. A correct protocol reaches a leaf
labelled $M(x,y)$ on every input pair.

The depth is the maximum number of transmitted bits on any
root-to-leaf path. We write $\D(M)$ for the minimum depth of a
correct protocol; a constant matrix therefore has complexity zero.
We call $\D(M)$ the communication complexity of $M$, or simply its
depth; a particular correct protocol may be deeper.
Either player may start, and successive bits may come from the same
player. The identity of the next speaker is determined by the protocol
and the preceding transcript, not by an uncommunicated choice. There
is no bound on the number of rounds. Local computation and the
length of a protocol's description are uncharged. An answer known
only to the last player is not yet a fixed leaf output: transmitting
it, when necessary, costs a bit.

For nonempty $U\subseteq X$ and $V\subseteq Y$, the notation
$M[U,V]$ means the submatrix on those original labels. Restricting
a correct protocol gives
\begin{equation}\label{eq:restriction}
 \D(M[U,V])\le\D(M).
\end{equation}
The reverse perspective will be useful: enlarging an already hard
submatrix with additional original inputs cannot make it easier.

All logarithms are to base two. Depth parameters and rank upper bounds
are integers. Rank means rank over $\R$ unless another field is
explicitly named. A \emph{row value} is the entire
Boolean vector indexed by the columns, as opposed to its input
label; column values are defined similarly.

We use $c$ for the communication complexity of the whole problem,
$s$ for a local protocol-depth bound, and $d$ for an algorithm's target
depth. Where a lower-bound promise $\D(M)\ge d$ is needed, we state it
explicitly. Additional parameters are defined where
they are needed.

\subsection{Repeated input values}

Identical column values are interchangeable for computing the
function. Choose one original representative of each value and
map every column label to its representative. A protocol on the
representatives gives a protocol on all columns: Bob applies the
map locally before following his prescribed messages. Conversely,
a protocol on all columns restricts to the representatives.
Thus deleting duplicate column values preserves communication
complexity and rank. The same holds for rows.

This observation concerns values, not the identity of a final
selected input. Whenever we retain representatives, they are
actual original labels. Whenever we later enlarge a selected
submatrix, we choose distinct unused original labels. These two
operations will never create a new row or column.

\section{Protocol domains}\label{sec:domains}

We want to count the column sets handled by short protocols. Even a
constant matrix has exponentially many column subsets, so we should not
count every cheap subset. Instead, fix Alice's instructions and include
\emph{every} column for which Bob can complete them correctly.
We call the resulting set a \emph{protocol domain}.

\subsection{Compatible columns}

There is an important restriction on what a completion can do.
In the lookup matrix \eqref{eq:lookup}, retain just rows $01$ and
$10$, and consider a one-bit protocol in which only Bob speaks,
with leaves labelled zero and one.
On column one the required answers are zero and
one, respectively. Bob cannot choose the right answer in both cases:
his input is the same. Column two fails for the same reason. Column
zero works, since its answer is zero on both rows.

Thus we must ask whether \emph{one} choice by Bob works on all the
retained rows, not whether each row could be handled by a separate
choice. The following definition imposes exactly this requirement.

Any protocol of depth at most $s$ can be padded to a complete binary
tree of depth $s$. Below a former leaf, add nodes transmitting a
fixed bit and give every new leaf the former output. This changes
neither correctness nor the depth bound.

\begin{definition}[Partial protocols and their domains]
\label{def:partial}
Fix a row set $U\subseteq X$ and a complete depth-$s$ tree.
A partial protocol $P$ specifies every speaker, every leaf output,
and Alice's message functions on $U$. Bob's message functions are
left unspecified.

A column $y\in Y$ is \emph{compatible} with $P$ on $U$ if there
is one assignment of bits to all Bob nodes that makes the tree
correct on every pair $(x,y)$ with $x\in U$.
The set of compatible columns, denoted by $G_P(U)$, is the
\emph{protocol domain} of $P$ on $U$.
\end{definition}

The assignment may depend on $y$, since Bob knows his own input.
It may not depend on $x$. If $U$ is empty, every column is
compatible; this convention is useful for empty branches of a tree.
Whenever we discuss the communication complexity of a restriction,
both its input sets are nonempty.

For the Bob-only protocol above, $G_P(\{01,10\})=\{0\}$.
Allowing separate choices for the two rows would incorrectly accept
all three columns. By contrast, the two depth-two protocols in the
introduction each handle all four rows, on their respective domains.

The point of compatibility is that it describes not just
individual successful executions but a domain with one legal
protocol.

\begin{lemma}[Compatibility and protocols]\label{lem:compatibility}
For nonempty $U\subseteq X$ and $V\subseteq Y$,
\[
 \D(M[U,V])\le s
 \quad\Longleftrightarrow\quad
 V\subseteq G_P(U)
 \text{ for some depth-$s$ partial protocol }P.
\]
If $G_P(U)$ is nonempty, the whole restriction
$M[U,G_P(U)]$ has a protocol of depth at most $s$.
\end{lemma}

\begin{proof}
Given a correct protocol on $U\times V$, pad it and forget
Bob's maps. Its original choices certify compatibility of every
column in $V$.

Conversely, for each $y\in G_P(U)$ choose one successful
assignment of bits to Bob nodes. At any fixed Bob node, these
chosen bits are now a function of $y$ alone. They therefore
define legal Bob message functions. Together with Alice's fixed
maps, they give a correct protocol on all of $U\times G_P(U)$.
The common assignment required in \cref{def:partial} is what
makes this conclusion valid.
\end{proof}

\subsection{The conditions at a leaf}

To describe $G_P(U)$, begin at the root with the rows in $U$.
At an Alice node, split the current row set according to her
fixed message function. At a Bob node, keep the current row set
unchanged when considering either child; Bob must decide which
child can serve that whole set.

Thus the rows associated with a node are those consistent with
Alice's choices along its root-to-node path. Earlier Bob choices
are ignored in forming this set. They determine which subtree
will be used, not which row Alice holds.

At a leaf labelled $b$, the accepted columns are exactly
\begin{equation}\label{eq:leaf-condition}
 \{y\in Y:M(x,y)=b\text{ for every associated row }x\}.
\end{equation}
At an Alice node, \emph{both} row groups must be served, so the
accepted columns are the intersection of the child sets.
At a Bob node, he may select \emph{one} child for all the rows,
so the accepted columns are their union.

The lookup example makes the difference visible. If Alice sends
$x_1$, her zero branch contains rows $00,01$, and her one branch
contains rows $10,11$. Column one gives the required answer in
both branches; column zero works only in the first. Now suppose
instead that Bob chooses whether Alice should send $x_1$ or $x_2$.
Each candidate protocol is tested on all four rows. He can select
the first on column one and the second on column two.
\Cref{fig:compatibility} computes the resulting column sets.

\begin{figure}[htbp]
\centering
\begin{tikzpicture}[x=1cm,y=1cm,>=Latex,
 every node/.style={font=\small},
 instruction/.style={draw,rounded corners=2pt,minimum width=2.65cm,
 minimum height=0.85cm,align=center},
 good/.style={draw=answergreen,fill=answergreen!10,text=answergreen,
 circle,minimum size=0.52cm,inner sep=0pt},
 bad/.style={draw=black!40,text=black!60,circle,
 minimum size=0.52cm,inner sep=0pt}]
\node at (0,2.7) {\textbf{An Alice node}};
\node[instruction,draw=aliceblue,fill=aliceblue!5] (a) at (0,1.9)
 {Alice sends $x_1$\\all four rows};
\node[instruction] (al) at (-1.65,0)
 {output $0$\\rows $00,01$};
\node[instruction] (ar) at (1.65,0)
 {output $1$\\rows $10,11$};
\draw[->] (a.south) -- node[pos=.6,above left=5pt,inner sep=0pt] {$0$} (al.north);
\draw[->] (a.south) -- node[pos=.6,above right=5pt,inner sep=0pt] {$1$} (ar.north);
\node[good] at (-2.2,-1.05) {$0$};
\node[good] at (-1.65,-1.05) {$1$};
\node[bad] at (-1.1,-1.05) {$2$};
\node[bad] at (1.1,-1.05) {$0$};
\node[good] at (1.65,-1.05) {$1$};
\node[bad] at (2.2,-1.05) {$2$};
\node[align=center] at (0,-1.8)
 {$\{0,1\}\cap\{1\}=\{1\}$\\Both row groups must be served.};

\node at (7.4,2.7) {\textbf{A Bob node}};
\node[instruction,draw=boborange,fill=boborange!5] (b) at (7.4,1.9)
 {Bob chooses a coordinate\\all four rows};
\node[instruction,draw=aliceblue,fill=aliceblue!5] (bl) at (5.75,0)
 {Alice sends $x_1$\\all four rows};
\node[instruction,draw=aliceblue,fill=aliceblue!5] (br) at (9.05,0)
 {Alice sends $x_2$\\all four rows};
\draw[->] (b.south) -- node[pos=.6,above left=5pt,inner sep=0pt] {$0$} (bl.north);
\draw[->] (b.south) -- node[pos=.6,above right=5pt,inner sep=0pt] {$1$} (br.north);
\node[bad] at (5.2,-1.05) {$0$};
\node[good] at (5.75,-1.05) {$1$};
\node[bad] at (6.3,-1.05) {$2$};
\node[bad] at (8.5,-1.05) {$0$};
\node[bad] at (9.05,-1.05) {$1$};
\node[good] at (9.6,-1.05) {$2$};
\node[align=center] at (7.4,-1.8)
 {$\{1\}\cup\{2\}=\{1,2\}$\\One candidate protocol suffices.};
\node[text=answergreen] at (3.7,-2.6)
 {A green circle marks a compatible column; numbers in circles are column labels.};
\end{tikzpicture}
\caption{Intersection and union in the same lookup matrix. On the
left, the leaves receive different row groups, and a column must
work on both. On the right, the children are complete one-bit Alice
protocols, each tested on all four rows; Bob chooses a successful
one using his column. Edge labels are transmitted bits. The rows
are split only at the Alice node, not at the Bob node.}
\label{fig:compatibility}
\end{figure}
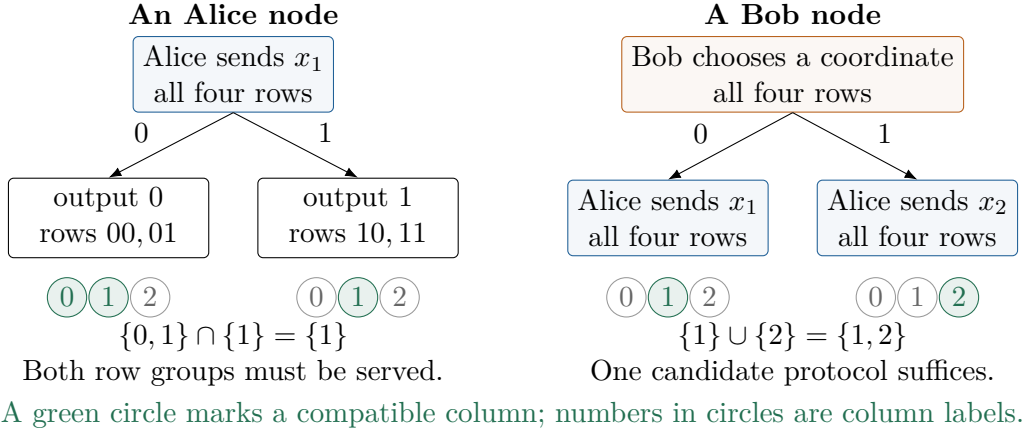
\FloatBarrier

Formally, this description follows by induction on the subtree
height. At an Alice node, successful Bob assignments in the two
children can be combined because they concern disjoint sets of
Bob nodes. At a Bob node, choose its bit and a successful
assignment below that child; decisions in the unused child can
be filled arbitrarily. At a leaf the requirement is precisely
\eqref{eq:leaf-condition}. Empty row groups satisfy every leaf
condition and cause no exception.

Thus every protocol domain is a depth-$s$ intersection/union tree
of leaf conditions. We next show that rank limits their number.

\section{Rank-sensitive counting}\label{sec:counting}

At a leaf, asking for the same answer on many rows can impose
far fewer distinct requirements than the number of row labels
suggests. In the lookup matrix \eqref{eq:lookup}, requiring answer zero on rows
$01$ and $10$ already excludes columns two and one. Adding row
$11$ excludes nothing further. The only accepted column is zero.

In general, for a matrix of positive
rank, an inclusion-minimal list of same-answer requirements has at most
as many members as the real rank.

\Needspace{10\baselineskip}
\subsection{Rank and input values}

We first collect the elementary facts we use: rank bounds the number
of distinct row and column values and hence the depth, depth bounds
rank, nonsingular minors give hard squares, and a hard matrix has many
rows and columns.

The lower bound by matrix rank is due to
Mehlhorn and Schmidt~\cite{MehlhornSchmidt82}. We include the elementary
proofs and the one-bit refinement needed for our fixed-leaf-output convention.

\begin{lemma}\label{lem:rank-facts}
Let $M$ be a finite total Boolean matrix of rank $r$.
\begin{enumerate}
\item There are at most $2^r$ distinct row values and at most
$2^r$ distinct column values. In particular, $\D(M)\le r+1$.
\item If $\D(M)\le s$ and $s\ge1$, then $r\le2^{s-1}$.
\item For each integer $1\le q\le r$, $M$ contains a nonsingular
$q$-by-$q$ submatrix on original indices. A nonsingular
$2^k$-by-$2^k$ Boolean matrix has depth exactly $k+1$ for $k\ge1$.
\item If $\D(M)\ge d\ge2$, each original input set has more than
$2^{d-2}$ elements.
\end{enumerate}
\end{lemma}

\begin{proof}
Choose $r$ basis rows. The entries of a column on these rows
determine all its other entries by linearity, so these $r$
Boolean coordinates distinguish all column values.
There are at most $2^r$ possibilities. Transpose to obtain
the row assertion. Naming a row value and then transmitting
the answer gives depth at most $r+1$.

For the second assertion, pad a depth-at-most-$s$ protocol to
a complete tree of depth $s$. Every transcript determines a
rectangle: Alice's decisions constrain only her input and
Bob's decisions constrain only his. At a node of depth $s-1$,
only one bit remains. Inside its rectangle, the one-output
inputs form either the empty set, the whole rectangle,
or one child rectangle. Their indicator matrix has rank
at most one. Summing these matrices over the $2^{s-1}$
nodes of that depth gives $M$, proving the rank bound.

The characterization of rank by nonsingular minors proves
the existence of the $q$-square. For $q=2^k$, the second
assertion rules out depth $k$; naming a row and transmitting
the answer gives the matching upper bound $k+1$.

Finally, if $m=\min\{\card X,\card Y\}$, the input-announcement
protocol gives
\begin{equation}\label{eq:dimension-bound}
 \D(M)\le\ceil{\log m}+1.
\end{equation}
If $m\le2^{d-2}$, this is at most $d-1$, contradicting the
hypothesis.
\end{proof}

The rank bound here includes the final answer convention.
The usual leaf-by-leaf decomposition would give the weaker
inequality $r\le2^s$. Grouping the last two leaves is what
provides the extra factor two used in our integer bounds.

\subsection{Short lists of row requirements}

We adapt a known counting method for unions of neighbourhoods in graph
algorithms~\cite[Proposition~3.6]{BuiXuanTelleVatshelle10}.
Oum, S{\ae}ther, and Vatshelle~\cite[Lemma~4.1 and Theorem~4.2]{OumSaetherVatshelle14}
count these unions by adding a short list of rows and recording the sum
on basis columns. Their rational rank equals real rank for Boolean
matrices. Below we give the argument for both possible leaf outputs.
The later passage from leaf conditions to domains of complete protocols
is the application needed for condensation.

Return to a leaf condition
\[
 M(x,y)=b\quad\text{for every }x\in U,
\]
where the required answer $b$ is fixed.
Remove redundant row requirements until the remaining list is
inclusion-minimal. Every row left on the list then has a witness:
some column fails that row's requirement but satisfies all
the other remaining requirements. \Cref{fig:leaf-witnesses} shows this
in the lookup matrix.

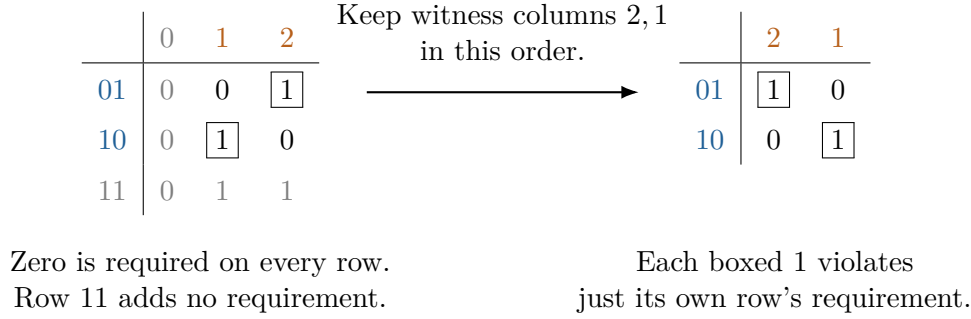
\begin{figure}[htbp]
\centering
\begin{tikzpicture}[x=1cm,y=1cm,>=Latex,
 every node/.style={font=\small}]
\node[anchor=north] (requirements) at (0,1.2) {%
 \renewcommand{\arraystretch}{1.4}%
 \begin{tabular}{c|ccc}
 &\textcolor{gray}{$0$}&\textcolor{boborange}{$1$}&\textcolor{boborange}{$2$}\\\hline
 \textcolor{aliceblue}{$01$}&\textcolor{gray}{0}&$0$&$\boxed{1}$\\
 \textcolor{aliceblue}{$10$}&\textcolor{gray}{0}&$\boxed{1}$&$0$\\
 \textcolor{gray}{$11$}&\textcolor{gray}{0}&\textcolor{gray}{1}&\textcolor{gray}{1}
 \end{tabular}};
\node[anchor=north] (witnesses) at (7.6,1.2) {%
 \renewcommand{\arraystretch}{1.4}%
 \begin{tabular}{c|cc}
 &\textcolor{boborange}{$2$}&\textcolor{boborange}{$1$}\\\hline
 \textcolor{aliceblue}{$01$}&$\boxed{1}$&$0$\\
 \textcolor{aliceblue}{$10$}&$0$&$\boxed{1}$
 \end{tabular}};
\draw[->,thick] (2.2,0) -- (5.8,0);
\node[align=center] at (4,0.8)
 {Keep witness columns $2,1$\\in this order.};
\node[anchor=north,align=center] at ([yshift=-2mm]requirements.south)
 {Zero is required on every row.\\Row $11$ adds no requirement.};
\node[anchor=north,align=center] at ([xshift=7.6cm,yshift=-2mm]requirements.south)
 {Each boxed $1$ violates\\just its own row's requirement.};
\end{tikzpicture}
\caption{Minimal leaf requirements in the lookup matrix.
Column two witnesses the necessity of row $01$; column one
witnesses the necessity of row $10$. Dropping the redundant row
and ordering the witness columns by their rows gives the identity
matrix on the right. The general proof below uses this same
witness pattern to bound the number of necessary requirements
by the rank.}
\label{fig:leaf-witnesses}
\end{figure}
\FloatBarrier

There is a shorter way to describe the requirements than naming all
the retained rows: add their entries. In the example, the two necessary
rows give
\[
 \underbrace{(0,0,1)}_{\text{row }01}
 +\underbrace{(0,1,0)}_{\text{row }10}=(0,1,1).
\]
A zero entry in this sum means that both rows answer zero, so the
accepted domain is exactly column zero. An entry equal to two would
mean that both rows answer one; here no column does. More generally,
for a sum of $q$ Boolean rows, we test whether the entry is zero or $q$,
according to the required answer. We need only the sum and, for answer
one, the number $q$ of rows. The proof now counts these descriptions
using basis columns.

\Needspace{8\baselineskip}
\begin{lemma}[Leaf conditions]\label{lem:leaf-count}
Suppose $\rank(M)\le r$ for an integer $r\ge1$.
Every same-answer leaf condition is equivalent, on the actual
column set, to at most $r$ distinct row-value requirements.
Including both possible outputs and all retained row sets,
the number of different leaf conditions is at most
\[
 2(r+1)^{r+1}.
\]
\end{lemma}

\begin{proof}
Let $S$ be an inclusion-minimal set of row values defining the
condition, and let $q=\card S$. For each $x\in S$, choose a
column $y_x$ that satisfies the other requirements and fails
the one at $x$.

If $b=0$, the entries on these rows and witness columns form
the identity matrix $I_q$. Hence $q\le r$. If $b=1$ and
$q\ge2$, they form $J_q-I_q$, where $J_q$ is the all-ones
matrix. The matrix $J_q-I_q$ has eigenvalues $q-1$ and
$-1$, so again the rank is $q$ and $q\le r$. The cases
$q=0$ and $q=1$ are covered by $r\ge1$. For $q\ge2$ the
witness columns are distinct, because a column failing one
selected requirement cannot witness failure of another
while satisfying the first.

Fix a set of basis columns of $M$; there are at most $r$ of them.
Let $z_S$ be the sum of the rows in $S$. Every other column is a
linear combination of the basis columns, so the entries of $z_S$
on the basis columns determine $z_S$ on every column. Each recorded
entry is an integer between zero and $q\le r$. Thus there are at most
$(r+1)^r$ possible sum vectors.

For output zero, the accepted columns are those where $z_S$ is zero,
giving at most $(r+1)^r$ different conditions. For output one, they
are those where $z_S$ equals $q$. Recording $q\in\{0,\ldots,r\}$
as well gives at most $(r+1)^{r+1}$ conditions. This includes $q=0$:
the empty sum is zero, and an empty list of requirements accepts
every column. Adding the two counts gives
\[
 (r+1)^r+(r+1)^{r+1}
 =(r+2)(r+1)^r\le2(r+1)^{r+1}.
\]
Different sums or row lists may define the same condition; we need
only an upper bound.
\end{proof}

Alice may send different messages on equal row values. The count still
applies separately at each leaf: it counts requirements, not arbitrary
row subsets.

\subsection{The domain count}

We can now describe an entire domain with few bits. At a leaf,
we record at most $r$ sum entries, a row count, and an output bit.
Each sum entry and the row count lie between zero and $r$,
so this takes $O(r\log(r+1))$ bits. A depth-$s$ tree has $2^s$
leaves and fewer internal nodes. Its leaf descriptions and
intersection/union choices therefore take
$O(r\log(r+1)2^s)$ bits in total.

Let $\mathcal C_s$ be the family of all sets $G_P(U)$ obtained
from depth-$s$ partial protocols, allowing both the initial
row set $U$ and Alice's maps to vary. These are subsets of the
original column set; the matrix is fixed while the retained rows
and Alice maps vary.

For the numerical bound, write
\begin{equation}\label{eq:counting-coefficient}
 B(r)=(r+1)\ceil{\log(r+1)}+2,\qquad r\ge1.
\end{equation}
This coefficient pays for the sum entries and row count at each leaf,
its output, and one intersection/union choice.

\begin{lemma}[Rank-sensitive domain count]\label{lem:domain-count}
If $\rank(M)\le r$ with $r\ge1$, then for every integer $s\ge0$,
\begin{equation}\label{eq:domain-count}
 \card{\mathcal C_s}
 \le 2^{B(r)2^s-1}.
\end{equation}
\end{lemma}

\begin{proof}
By the recursive description in \cref{sec:domains},
a protocol domain is an intersection/union formula on
a complete depth-$s$ tree. There are $2^s-1$ internal nodes,
each choosing one of those two operations, and $2^s$ leaves.
\Cref{lem:leaf-count} bounds the number of conditions at each
leaf. Since $r+1\le2^{\ceil{\log(r+1)}}$, the number of possible
formulas is at most
\[
 2^{2^s-1}
 \left(2(r+1)^{r+1}\right)^{2^s}
 \le 2^{B(r)2^s-1}.
\]
Varying the initial rows and Alice maps changes the leaf
conditions but does not create a condition outside the
counted family. We have also counted formulas whose leaf
conditions may not come from a legal partial protocol.
This overcounting is harmless.
\end{proof}

The count is independent of the original input-set sizes. It counts
domains, not protocol descriptions, and does not yet enumerate them;
\cref{sec:construction} supplies that algorithm.

\section{Lossless condensation}\label{sec:condensation}

We turn the domain count into a cover. If every small column restriction
is easy, a few cheap domains cover all columns; Bob names one and runs
its protocol. The contrapositive yields a small hard column set.
We then transpose to restrict the rows.

\subsection{Hitting sets and protocol covers}

Before introducing the notation, consider a remaining column set $W$.
Suppose every short-protocol domain covers at most half of $W$.
A list of $t$ independently chosen columns lies entirely inside any
one such domain with probability at most $2^{-t}$. If there are fewer
than $2^t$ domains, some list lies inside none of them. But if every
restriction to at most $t$ columns has a short protocol, that is
impossible: its protocol domain contains the list.
Therefore one domain must cover more than half of $W$.

The next lemma turns this argument into a small test set.
The finite covering principle is adapted from
Hrube\v{s}~\cite[Lemma~3.1]{Hrubes25}; here the covering sets are
domains of whole protocols.

We need a set of columns meeting every protocol-domain complement
that has probability mass at least $1/2$. Such a set is a
\emph{hitting set} for these complements. Meeting a complement means
finding a column on which that partial protocol cannot be completed.

\Needspace{13\baselineskip}
\begin{lemma}[A finite hitting set]\label{lem:finite-hit}
Let $\rank(M)\le r$ with $r\ge1$, and let $s\ge0$.
Under any distribution on the columns, there is a set of size at most
\begin{equation}\label{eq:finite-sample}
 a(r,s)=B(r)2^s,
 \qquad B(r)=(r+1)\ceil{\log(r+1)}+2,
\end{equation}
meeting every complement $Y\setminus G$ with $G\in\mathcal C_s$
and probability mass at least $1/2$.
The set can be chosen nonempty and uses original column labels.
\end{lemma}

\begin{proof}
Draw $a(r,s)$ columns independently. A fixed set of mass
at least one half is missed with probability at most
$2^{-a(r,s)}$. By \cref{lem:domain-count}, a union bound over all
complements gives failure probability at most $1/2$.
A successful sample therefore exists. Its distinct support is
nonempty and has no larger size.
\end{proof}

The distribution may be chosen afresh on any remaining set of columns.
We will use the uniform distribution on the part not yet covered;
\cref{fig:cover-progress} shows one round of this covering procedure.
This freedom is why a hitting-set statement for every distribution,
rather than for one fixed original distribution, is useful.

\begin{figure}[tbp]
\centering
\begin{tikzpicture}[x=.67cm,y=1cm,>=Latex,font=\small]
  \node[anchor=east,align=right] at (-.65,1.25) {remaining\\columns $W$};
  \node[anchor=east,align=right] at (-.65,0) {one protocol\\domain $G$};
  \node[anchor=east,align=right] at (-.65,-1.25) {next remainder\\$W\setminus G$};
  \foreach \i in {0,...,11} {
    \node[circle,draw=black!50,fill=black!6,minimum size=4.5mm,inner sep=0pt] at (\i,1.25) {};
    \node[circle,draw=black!35,fill=black!6,minimum size=4.5mm,inner sep=0pt] at (\i,0) {};
  }
  \foreach \i in {0,...,6} {
    \node[circle,draw=answergreen,fill=answergreen!25,minimum size=4.5mm,inner sep=0pt] at (\i,0) {};
  }
  \foreach \i in {1,3,5} {
    \node[circle,draw=boborange,line width=1.1pt,minimum size=6.5mm,inner sep=0pt] at (\i,1.25) {};
    \node[circle,draw=boborange,line width=1.1pt,minimum size=6.5mm,inner sep=0pt] at (\i,0) {};
  }
  \foreach \i in {7,...,11} {
    \node[circle,draw=black!60,fill=black!12,minimum size=4.5mm,inner sep=0pt] at (\i,-1.25) {};
    \draw[->,black!45] (\i,-.3) -- (\i,-.9);
  }
  \node[anchor=west,align=left] at (12,1.25) {circled points:\\test set $C$};
  \node[anchor=west,align=left] at (12,0) {$C\subseteq G$\\by local easiness};
  \node[anchor=west,align=left] at (12,-1.25) {fewer than\\$|W|/2$ remain};
\end{tikzpicture}
\caption{Why a small test set forces progress. Choose $C$ to meet every
protocol-domain complement containing at least half of $W$. A short
protocol on $C$ extends to a protocol domain $G$ containing $C$.
Its complement misses $C$, so it must contain less than half of $W$.
Repeat with a \emph{new} test set on the remainder. Schematic illustration
of one round.}
\label{fig:cover-progress}
\end{figure}
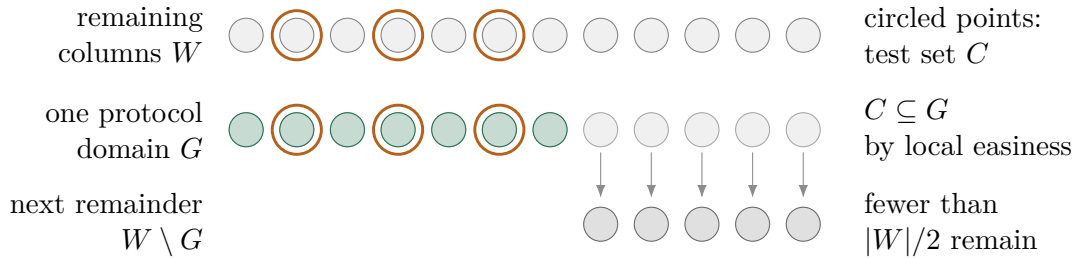

We use $a(r,s)=B(r)2^s$ throughout. Since $\mathcal C_s$ includes all
retained row sets, the same bound remains available after a row restriction.

\begin{lemma}[A cover by protocols]\label{lem:cover}
Suppose $\rank(M)\le r$ for an integer $r\ge1$.
If every nonempty column set $C$ with
$\card C\le a(r,s)$ satisfies $\D(M[X,C])\le s$, then
\begin{equation}\label{eq:cover-cost}
 \D(M)\le s+\ceil{\log r}.
\end{equation}
\end{lemma}

\begin{proof}
Keep one representative of each column value. There are at
most $2^r$ representatives by \cref{lem:rank-facts}, and
the local hypothesis remains true.

Let $W$ be any nonempty set of representatives not yet
covered. Under its uniform distribution, choose a hitting set
$C\subseteq W$ of size at most $a(r,s)$. The local
hypothesis and \cref{lem:compatibility} give a partial
protocol $P$ with $C\subseteq G_P(X)$.
If $W\setminus G_P(X)$ had at least half the elements of $W$,
the hitting set would meet that complement, contradicting
its containment in $G_P(X)$. Thus this domain covers
strictly more than half of $W$.

Repeat on the uncovered representatives. At most $r$ choices
suffice: after $r$ strict halvings the number remaining is
strictly less than $2^r/2^r=1$, hence zero. For each chosen
domain, \cref{lem:compatibility} supplies a depth-$s$
protocol valid on all original rows.

Bob chooses the first domain containing his column value
and announces its index in $\ceil{\log r}$ bits. The players
then run its fixed protocol. The choice depends only on
Bob's input, as in \cref{fig:domain-cover}. Restoring duplicate
column labels by the representative map costs no communication.
\end{proof}

The covering domains need not be disjoint or monochromatic;
each already carries a correct protocol on the whole row set.

\subsection{Both input sets}

The contrapositive of \cref{lem:cover} finds a small hard column
set while leaving all rows available. To find a small hard row
set as well, we apply it again after transposing. The first
restriction must retain enough depth to pay for the second
application, as \cref{fig:two-sided} shows. In \cref{lem:two-sided},
$H$ bounds the number of bits needed to name a covering domain; this
cost is paid once in each orientation.

\begin{figure}[htbp]
\centering
\begin{tikzpicture}[x=1cm,y=1cm,>=Latex,font=\small]
  \foreach \start in {0,5,10} {
    \fill[black!3] (\start,0) rectangle (\start+2.4,2.4);
  }
  \foreach \i in {0,2,5} {
    \fill[boborange!22] (5+.4*\i,0) rectangle (5+.4*\i+.4,2.4);
    \foreach \j in {1,3,4} {
      \fill[answergreen!28] (10+.4*\i,.4*\j) rectangle (10+.4*\i+.4,.4*\j+.4);
    }
  }
  \foreach \start in {0,5,10} {
    \foreach \i in {1,...,5} {
      \draw[black!25,thin] ({\start+.4*\i},0) -- ({\start+.4*\i},2.4);
      \draw[black!25,thin] (\start,{.4*\i}) -- ({\start+2.4},{.4*\i});
    }
    \draw[black!45] (\start,0) rectangle ({\start+2.4},2.4);
  }
  \node at (1.2,2.85) {$M[X,Y]$};
  \node at (6.2,2.85) {$M[X,V]$};
  \node at (11.2,2.85) {$M[U,V]$};
  \draw[->,thick] (2.75,1.2) -- (4.65,1.2)
    node[midway,above,align=center] {restrict\\columns};
  \draw[->,thick] (7.75,1.2) -- (9.65,1.2)
    node[midway,above,align=center] {restrict\\rows};
  \node at (1.2,-.45) {depth $>s+2H$};
  \node at (6.2,-.45) {depth $>s+H$};
  \node at (11.2,-.45) {depth $>s$};
  \node[align=center] at (6.2,-1.1) {$|V|\le a(r,s+H)$\\all rows remain};
  \node[align=center] at (11.2,-1.1) {$|U|\le a(r,s)$\\keep the same $V$};
\end{tikzpicture}
\caption{Two successive restrictions. First the column set $V$
(orange) is chosen with every row kept; then the row set $U$ is chosen
inside $M[X,V]$, and $M[U,V]$ is shown in green. Each use of the cover
lemma loses at most $\lceil\log r\rceil\le H$ bits of depth, so the
first restriction must leave an extra $H$ bits for the second. Grid
sizes are illustrative.}
\label{fig:two-sided}
\end{figure}
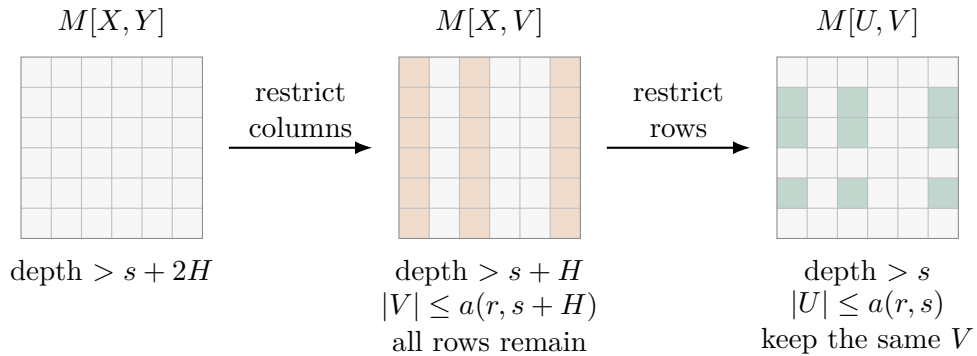

\Needspace{11\baselineskip}
\begin{lemma}[A two-sided restriction]\label{lem:two-sided}
Let $\rank(M)\le r$ with $r\ge1$. Let $H$ be any integer
with $H\ge\ceil{\log r}$, and let $s\ge0$.
If $\D(M)>s+2H$, then $M$ has an original submatrix $M[U,V]$
such that
\begin{equation}\label{eq:two-sided}
 \D(M[U,V])>s,\qquad
 \card U\le a(r,s),\qquad \card V\le a(r,s+H).
\end{equation}
In particular, this holds with $H=\ceil{\log(r+1)}$.
\end{lemma}

\begin{proof}
The contrapositive of \cref{lem:cover}, used at depth $s+H$,
gives a column set $V$ with at most $a(r,s+H)$ elements and
\[
 \D(M[X,V])>s+H.
\]
Indeed, otherwise the whole matrix would have depth at most
$s+H+\ceil{\log r}\le s+2H$.

Transpose this restricted matrix. Its rank is still at most
$r$, and its depth is still greater than
$s+\ceil{\log r}$. The same contrapositive at depth $s$
therefore supplies a set $U$ of at most $a(r,s)$ original rows
with $\D(M[U,V])>s$. These are successive restrictions of
the original matrix, not independently selected rows and columns.
\end{proof}

\subsection{A rank-sensitive square}

Write
\[
 c=\D(M),\qquad r=\rank(M),\qquad h=\ceil{\log(r+1)}.
\]
The choice $H=h$ pays for the cover index and also gives $r+1\le2^h$.
The row-sum count has the useful bound
\[
 B(r)=(r+1)h+2\le(r+1)(h+1)\le2^h(h+1).
\]
Consequently, the larger side in \cref{lem:two-sided} has at most
$(h+1)2^{s+2h}$ inputs. In this size bound, one factor $2^h$ comes from
the first restriction's depth budget and the other from the rank-dependent count.
The remaining factor $h+1$ costs only $\ceil{\log(h+1)}$ bits.

\begin{theorem}[Rank-sensitive condensation]\label{thm:rank-core}
Let $c=\D(M)\ge2$, $r=\rank(M)$, and $h=\ceil{\log(r+1)}$.
Then $M$ contains an original
$2^{c-2}$-by-$2^{c-2}$ submatrix $N$ with
\[
 \D(N)\ge c-2h-\ceil{\log(h+1)}-1.
\]
\end{theorem}

\begin{proof}
Suppose first that the stated lower bound is positive.
Put $\tau=\ceil{\log(h+1)}$ and set $s=c-2h-\tau-2\ge0$.
Then $c>s+2h$, so \cref{lem:two-sided} with $H=h$
gives a restriction of depth at least $s+1=c-2h-\tau-1$.
Its larger side has at most
\[
 a(r,s+h)
 =B(r)2^{s+h}
 \le(h+1)2^{s+2h}
 =(h+1)2^{c-\tau-2}\le2^{c-2}
\]
elements, because $h+1\le2^\tau$.

By \cref{lem:rank-facts}, each original input set has more
than $2^{c-2}$ elements. Enlarge both selected sets to
exactly that size, using distinct original labels.
The depth cannot decrease under enlargement.

If the stated depth lower bound is nonpositive, any original
square of that size suffices; the same dimension argument
guarantees that one exists.
\end{proof}

A matrix with small logarithmic rank therefore has a square
retaining nearly all of its depth. A matrix with large rank
has a different source of hardness: a nonsingular minor.
First we combine these two sources of hardness at input length $c-2$,
proving \cref{cor:full-input-length}.

\begin{proof}[Proof of \cref{cor:full-input-length}]
Write $r=\rank(M)$, $h=\ceil{\log(r+1)}$, and
$\ell=\ceil{\log(c+1)}$.
The rank bounds in \cref{lem:rank-facts} give $r\ge c-1\ge3$ and
$2\le h\le c$. A nonsingular minor of order
$2^{\min\{h-1,c-2\}}$ has depth $\min\{h,c-1\}$.
Enlarge it to a square of side $2^{c-2}$ using original inputs;
\cref{lem:rank-facts} ensures that enough labels are available.
Its depth cannot decrease.

On the other hand, \cref{thm:rank-core} supplies a square of that same
size and depth at least $c-2h-\ell-1$, since
$\ceil{\log(h+1)}\le\ell$.
If $h\le c-1$, the harder of the two squares therefore has depth at least
\[
 \max\{h,\,c-2h-\ell-1\}\ge\frac{c-\ell-1}{3}.
\]
Indeed, either $h$ reaches the right-hand side or the second term does.
If $h=c$, the padded minor already has depth $c-1$, which is stronger.
\end{proof}

We now allow the retained input length to vary so that the square can
be nearly as hard as its size permits. The same rank split gives the
following quantitative form of the first main theorem.

\begin{theorem}[Nearly maximal condensation]\label{thm:near-maximal}
Let $\D(M)=c\ge4$, let $0<\varepsilon<1$, and put
$\ell=\ceil{\log(c+1)}$.
There are an integer $1\le k\le c-1$ and an original $2^k$-by-$2^k$
submatrix $N$ such that
\[
 \D(N)\ge(1-\varepsilon)(k+1)
 \quad\text{and}\quad
 k\ge\max\left\{\frac{\varepsilon c-4}{3},\,
                  \frac{\varepsilon c-\ell-3}{2}\right\}.
\]
\end{theorem}

\begin{proof}
Write $r=\rank(M)$, $h=\ceil{\log(r+1)}$ and $j=\floor{\log r}$.
Since $r$ is a positive integer,
$h=j+1$. Moreover, $c\le r+1$ and $c\ge4$ imply $r\ge3$,
so $j\ge1$; the rank upper bound in \cref{lem:rank-facts}
gives $j\le c-1$.

Write $q=\max\{(\varepsilon c-4)/3,(\varepsilon c-\ell-3)/2\}$.
If $j\ge q$, choose a nonsingular minor of order $2^j$. Its depth is
exactly $j+1$, by \cref{lem:rank-facts}. Taking $k=j$
meets both requirements.

Otherwise $j<q$. Since $h=j+1\le c$, the integer
$\tau=\ceil{\log(h+1)}$ satisfies both $\tau\le h=j+1$ and $\tau\le\ell$.
\Cref{thm:rank-core} gives a square with $k=c-2$ and
\[
 \D(N)\ge c-2j-\tau-3
 \ge c-\min\{3j+4,\,2j+\ell+3\}.
\]
At least one of the two terms in this minimum is less than
$\varepsilon c$, by $j<q$. Thus $\D(N)>(1-\varepsilon)c$, which
exceeds $(1-\varepsilon)(k+1)$. Also $q<c/2\le c-2$ for $c\ge4$,
so $k\ge q$ and $1\le k\le c-1$.
\end{proof}

This proves \cref{thm:condensation}: high rank supplies a maximally
hard minor, while low rank supplies a larger square losing only a
controlled number of bits. We now make the search explicit by finding
the counted domains together with the protocols they carry.

\section{Constructive condensation}
\label{sec:construction}

The existence proof counts protocol domains and uses a small cover to
assemble their protocols. We now make those steps effective. The running
time depends polynomially on the original table size. The exponent is
independent of the target depth.

The additional task is not just to list sets: every listed domain must
carry a short protocol, and every cheap restriction must be contained in
a listed domain. We meet both requirements by tracking the rows and
columns on which each subtree works. The same covering argument then
finds hard inputs without computing the optimal communication depth of
the original matrix.

Throughout this section, $T=|X||Y|$ is the explicit table size. Running
time refers to computation on that table, not to local computation within
a communication protocol. All finite lists have fixed orders, so choices
can be made deterministically by taking the first permitted object.

\subsection{Leaf domains}

Recall the lookup matrix \eqref{eq:lookup}: Alice holds
$x=x_1x_2\in\{00,01,10,11\}$, and Bob holds $y\in\{0,1,2\}$.
Suppose a leaf must answer zero on rows $00$ and $01$. Row $00$ places
no restriction on Bob's query; row $01$ excludes query two. The accepted
columns are therefore $\{0,1\}$. Adding row $10$ leaves only column zero.
Thus adding one row requirement means intersecting the current domain
with the columns on which that row has the required answer.

This gives a direct way to list all possible \emph{leaf domains},
the column sets accepted at a leaf as in \eqref{eq:leaf-condition}. For
$b\in\{0,1\}$, write
\[
 A_b(x)=\{y\in Y:M(x,y)=b\}.
\]
Start with the domain $Y$, corresponding to no row requirements.
Process the rows in a fixed order. At row $x$, keep every domain already
on the list and also its intersection with $A_b(x)$; remove duplicate
sets. After any initial group of rows, the list contains exactly the
intersections obtained by selecting a subset of those rows. In particular,
the final list contains every same-answer leaf domain and nothing else.
This is the elementary intersection-enumeration procedure; see
Mary and Strozecki~\cite{MaryStrozecki19} for the broader enumeration setting.

\begin{lemma}[Listing leaf domains]\label{lem:leaf-list}
Let $\rank(M)\le r$ for an integer $r\ge1$. All pairs $(b,C)$ in which
$b\in\{0,1\}$ and $C$ is a same-answer leaf domain can be listed in time
$2^{O(r\log(r+1))}T^{O(1)}$. There are at most
$2(r+1)^{r+1}$ such pairs.
\end{lemma}

\begin{proof}
Use the intersection procedure just described, separately for each
answer. Every intermediate domain is also a possible final domain:
simply add no further requirements. Thus no intermediate list is longer
than the final list. Lemma~\ref{lem:leaf-count} bounds the combined number
of output-labelled domains by $2(r+1)^{r+1}$. Representing sets by their
membership strings, each intersection and equality test costs polynomial
time in $T$. Even pairwise duplicate removal at every step fits the
stated bound.
\end{proof}

Thus the algorithm uses intersections to list domains; the row-sum
argument from \cref{sec:counting} bounds the list's length.

\subsection{Protocols on rectangles}

A proposed leaf domain tells us which columns we would like to accept.
We must also determine which rows can safely reach that leaf. For a
leaf labelled $b$ with proposed column set $C$, these are exactly
\[
 R=\{x\in X:M(x,y)=b\text{ for every }y\in C\}.
\]
The leaf is therefore correct throughout the rectangle $R\times C$.
If $C$ is empty, every row qualifies. Such empty rectangles are harmless
while assembling a tree; the eventual output restriction will be nonempty.

Return to the lookup example. The zero leaf with column domain
$\{0,1\}$ accepts rows $00,01$. The one leaf with column domain
$\{1\}$ accepts rows $10,11$. Put an Alice node above these leaves.
Alice can choose the correct child for every row, so the parent accepts
the union of their row sets. Bob's query must work whichever row Alice
has, so the parent accepts the intersection of their column sets:
all four rows, but only column one. Alice's message is precisely $x_1$.

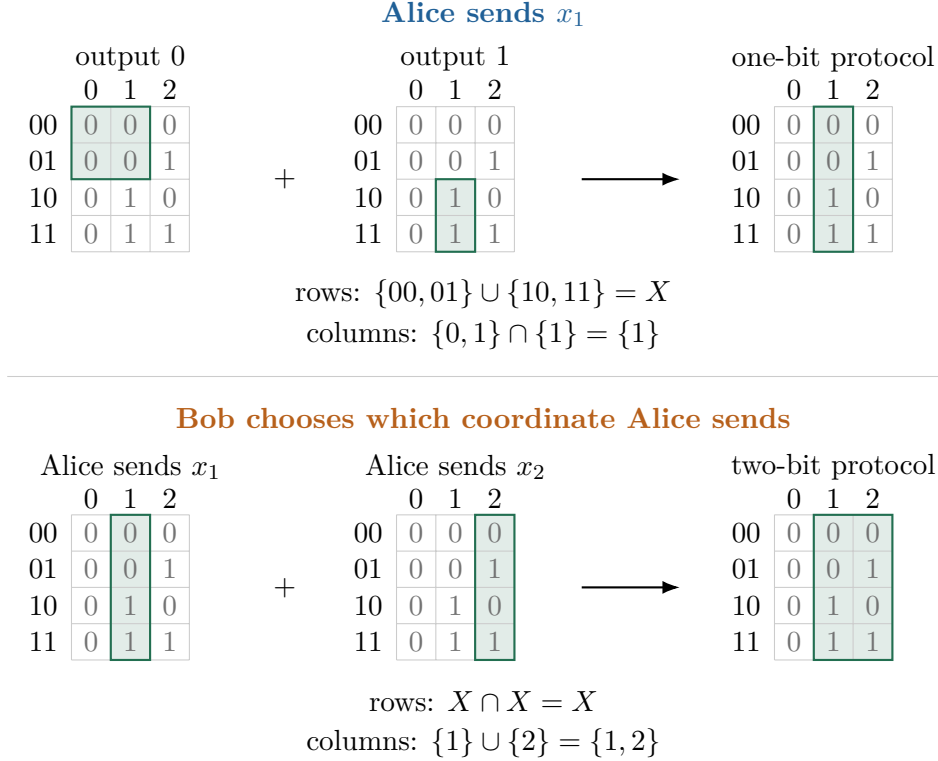
\begin{figure}[!t]
\centering
\begingroup
\newcommand{\lookuprectangle}[2]{%
  \begin{scope}[shift={#1},x=.52cm,y=.48cm]
    \foreach \left/\bottom/\right/\top in {#2}{
      \fill[answergreen!13] (\left,\bottom) rectangle (\right,\top);
    }
    \draw[step=1,black!20] (0,0) grid (3,4);
    \foreach \j in {0,1,2}
      \node at (\j+.5,4.45) {$\j$};
    \foreach \i/\row in {3/00,2/01,1/10,0/11}
      \node[anchor=east] at (-.2,\i+.5) {$\row$};
    \foreach \j/\i/\value in
      {0/3/0,1/3/0,2/3/0,0/2/0,1/2/0,2/2/1,
       0/1/0,1/1/1,2/1/0,0/0/0,1/0/1,2/0/1}
      \node[text=black!55] at (\j+.5,\i+.5) {$\value$};
    \foreach \left/\bottom/\right/\top in {#2}{
      \draw[answergreen,line width=.8pt]
        (\left,\bottom) rectangle (\right,\top);
    }
  \end{scope}%
}
\begin{tikzpicture}[x=1cm,y=1cm,>=Latex,
  every node/.style={font=\normalsize,inner sep=2pt}]
  \node[text=aliceblue] at (5.45,3.15)
    {\textbf{Alice sends $x_1$}};
  \node at (.78,2.55) {output $0$};
  \node at (5.08,2.55) {output $1$};
  \node at (10.08,2.55) {one-bit protocol};
  \lookuprectangle{(0,0)}{0/2/2/4}
  \lookuprectangle{(4.3,0)}{1/0/2/2}
  \lookuprectangle{(9.3,0)}{1/0/2/4}
  \node at (2.82,.96) {$+$};
  \draw[->,line width=.8pt] (6.75,.96) -- (8.05,.96);
  \node at (5.45,-.55)
    {rows: $\{00,01\}\cup\{10,11\}=X$};
  \node at (5.45,-1.10)
    {columns: $\{0,1\}\cap\{1\}=\{1\}$};

  \draw[black!20] (-.85,-1.65) -- (11.5,-1.65);
  \node[text=boborange] at (5.45,-2.20)
    {\textbf{Bob chooses which coordinate Alice sends}};
  \node at (.78,-2.85) {Alice sends $x_1$};
  \node at (5.08,-2.85) {Alice sends $x_2$};
  \node at (10.08,-2.85) {two-bit protocol};
  \lookuprectangle{(0,-5.4)}{1/0/2/4}
  \lookuprectangle{(4.3,-5.4)}{2/0/3/4}
  \lookuprectangle{(9.3,-5.4)}{1/0/3/4}
  \node at (2.82,-4.44) {$+$};
  \draw[->,line width=.8pt] (6.75,-4.44) -- (8.05,-4.44);
  \node at (5.45,-5.95) {rows: $X\cap X=X$};
  \node at (5.45,-6.50)
    {columns: $\{1\}\cup\{2\}=\{1,2\}$};
\end{tikzpicture}
\endgroup
\caption{Building protocols from rectangles in the lookup matrix
\eqref{eq:lookup}; shaded entries are the supported inputs and
$X=\{00,01,10,11\}$. Above, each child is a constant-output rectangle.
Alice uses her row to choose a child, so the rows combine by union,
but a column must work in both children. Below, each child is already
a one-bit protocol on all four rows. Bob uses his column to choose
a child, so the columns combine by union, but a row must work in
both children. The parent rectangles need not be monochromatic:
they come with the composed protocols.}
\label{fig:rectangle-recursion}
\end{figure}

There is an analogous protocol for column two, in which Alice sends
$x_2$. Bob can select between these two protocols. Now the roles of
union and intersection are reversed: both children must accept Alice's
row, but only the chosen child must accept Bob's column. The resulting
protocol accepts all four rows and columns $\{1,2\}$, as shown in
\cref{fig:rectangle-recursion}.

We record these two rules for a general tree. At node $v$, the sets
$R_v$ and $C_v$ describe the rectangle on which the subtree will work.
They are computed from the children, rather than supplied as a protocol
in advance.

\Needspace{14\baselineskip}
\begin{lemma}[Protocols from rectangles]\label{lem:rectangle-recursion}
Fix the speakers of a complete depth-$s$ binary tree, its leaf outputs,
and a proposed column set $C_v$ at each leaf $v$. At a leaf of output $b_v$,
put
\[
 R_v=\{x\in X:M(x,y)=b_v\text{ for every }y\in C_v\}.
\]
At each internal node with children $v_0,v_1$, use the rules
\[
\begin{array}{c|cc}
 \text{speaker at }v&R_v&C_v\\ \hline
 \text{Alice}&R_{v_0}\cup R_{v_1}&C_{v_0}\cap C_{v_1}\\
 \text{Bob}&R_{v_0}\cap R_{v_1}&C_{v_0}\cup C_{v_1}.
\end{array}
\]
These sets determine legal message functions making the subtree correct
on $R_v\times C_v$. All sets and message functions can be computed in
$2^{O(s)}T^{O(1)}$ time.
\end{lemma}

\begin{proof}
At a leaf, correctness is the definition of $R_v$. At an Alice node,
she chooses the first child whose row set contains her input. A row
in $R_v$ has such a child, and every column in $C_v$ belongs to both
child column sets. At a Bob node, he chooses the first child whose
column set contains his input. A column in $C_v$ has such a child,
and every row in $R_v$ belongs to both child row sets. Each message
depends only on the speaker's input. Induction on the subtree height
proves correctness. Give arbitrary messages, say zero, to inputs outside
the relevant union; these choices are never needed on $R_v\times C_v$.

There are $2^s$ leaves and fewer internal nodes. At a leaf, scanning
the table finds $R_v$; at an internal node, the two set operations and
the message choices are scans of membership strings. This proves the
running-time bound, including the stored message functions.
\end{proof}

We can now test the formulas counted in \cref{lem:domain-count}.
Enumerate the speakers and the output-labelled leaf domains, compute
the rectangle at the root, and keep the tree when its row set is all
of $X$. Lemma~\ref{lem:rectangle-recursion} then supplies a protocol on
every original row and the computed column domain. The next theorem
shows that this test retains a domain containing every cheap restriction.

Recall the size bound
\[
 a(r,s)=B(r)2^s,\qquad
 B(r)=(r+1)\ceil{\log(r+1)}+2.
\]
This is the same size bound used in \cref{lem:finite-hit};
here it also controls the logarithm of the protocol-domain list size.

\begin{theorem}[A list of protocol domains]\label{thm:protocol-list}
Let $\rank(M)\le r$ with integer $r\ge1$, and let $s\ge0$ be an integer.
In deterministic time
\[
 2^{O(B(r)2^s)}T^{O(1)},
\]
one can construct a list $\mathcal G$ of protocol domains,
each nonempty member supplied with an actual depth-at-most-$s$ protocol,
such that, for every nonempty $C\subseteq Y$,
\begin{equation}\label{eq:list-containment}
 \D(M[X,C])\le s
 \quad\Longrightarrow\quad
 C\subseteq G\text{ for some }G\in\mathcal G.
\end{equation}
The list has at most $2^{a(r,s)-1}$ distinct members.
\end{theorem}

\begin{proof}
Use Lemma~\ref{lem:leaf-list} to list the output-labelled leaf domains.
Enumerate every speaker assignment on a complete depth-$s$ tree and
every choice from this list at each leaf. For each resulting tree, apply
Lemma~\ref{lem:rectangle-recursion}. Discard it unless $R_{\rm root}=X$.
For a retained tree, fix the constructed Alice maps and compute their
full protocol domain by the recursion in \cref{sec:domains}.
Store that domain and the corresponding Bob maps, keeping one record
per distinct domain. Full completion can only enlarge $C_{\rm root}$,
since the already constructed Bob maps work there. It takes another
table scan per leaf and set operation per node, hence
$2^{O(s)}T^{O(1)}$ time.

Every stored protocol is valid by Lemma~\ref{lem:rectangle-recursion}
and full completion. To prove coverage, start with a protocol of depth
at most $s$ on $X\times C$, pad it, and forget Bob's maps. At each node
$v$, let $U_v$ be the rows consistent with Alice's decisions on the path
to $v$, ignoring Bob's decisions. At a leaf, choose the actual domain
\[
 C_v=\{y:M(x,y)=b_v\text{ for every }x\in U_v\}.
\]
It occurs in the leaf list. Thus this speaker tree and these leaf
choices occur in the enumeration.

For this choice, $U_v\subseteq R_v$ at every node. This is immediate
at a leaf. At an Alice node, $U_v$ is partitioned between its two
children, so their containment assertions imply containment in
$R_{v_0}\cup R_{v_1}$. At a Bob node, both children have the same
row set $U_v$, so containment holds in their intersection. At the root,
$U_{\rm root}=X$, and the tree passes the test. Its recursively computed
column domain is the full protocol domain of the original partial
protocol, by \cref{sec:domains}, and therefore contains $C$. The final
completion of the newly chosen Alice maps preserves this containment.
This proves \eqref{eq:list-containment} even though those maps need not
be the original ones.

Finally, put $L=2^s$. There are $2^{L-1}$ speaker assignments and at most
$2(r+1)^{r+1}$ choices per leaf. Exactly as in
Lemma~\ref{lem:domain-count}, the number of trials is at most
\[
 2^{L-1}\bigl(2(r+1)^{r+1}\bigr)^L
 \le 2^{B(r)L-1}=2^{a(r,s)-1}.
\]
Each trial has polynomial table cost times $2^{O(s)}$. Listing the leaves,
and even pairwise deduplication of the completed domains, fit within
$2^{O(a(r,s))}T^{O(1)}$. Since each trial produces at most one domain,
the stated list-size bound also follows. A depth-zero tree has one leaf
and no speaker choices, so the argument includes $s=0$.
\end{proof}

The two parts of this argument serve different purposes. The rectangle
test certifies that a listed formula yields legal messages. The
containment proof certifies that the enumeration does not miss a cheap
restriction. A count alone would establish neither property.

\subsection{Protocol covers and hard column sets}

With the finite list in hand, the covering proof becomes constructive.
First try to cover the inputs: repeatedly use a listed domain containing
more than half the remaining columns. If this succeeds, Bob can name
a protocol that handles his input. If it stops, select columns that
eliminate at least half the protocols still able to handle all columns
selected so far. When no listed protocol remains, the containment
property certifies that the selected restriction is hard.

The two phases halve different objects: first the uncovered inputs,
then the possible protocols. Their opposite membership tests are the
algorithmic form of the cover and hitting-set arguments in
\cref{sec:condensation}.

\begin{lemma}[A constructive one-sided alternative]
\label{lem:constructive-one-sided}
Suppose $\rank(M)\le r$ with integer $r\ge1$, let $s\ge0$ be an integer, and put
$\lambda=\ceil{\log r}$. In time $2^{O(B(r)2^s)}T^{O(1)}$, a deterministic
algorithm returns either a protocol for all of $M$ of depth at most
$s+\lambda$, or a nonempty set of original columns $C$ with
\[
 |C|\le a(r,s),\qquad \D(M[X,C])>s.
\]
\end{lemma}

\begin{proof}
Keep one original representative of each column value, recording the
map from all original columns to representatives. There are at most
$2^r$ representatives by Lemma~\ref{lem:rank-facts}. Construct the list
$\mathcal G$ on this matrix and discard empty domains. If the list is
empty, choose any representative column. Its singleton restriction
has depth greater than $s$ by \eqref{eq:list-containment}, and its size
is at most $a(r,s)$. This case can occur when $s=0$.

Otherwise let $W$ initially be the representative universe. Whenever
a listed domain contains strictly more than half of $W$, record its
protocol and remove its columns from $W$. If this empties $W$, at most
$r$ domains have been recorded: after $r$ strict halvings of a set of
size at most $2^r$, fewer than one element remains. Bob chooses the
first recorded domain containing his input, announces its index in
$\lambda$ bits, and executes the stored protocol. Restoring equal
columns by the representative map costs no communication.

If the covering process stops with $W\ne\varnothing$, then
\begin{equation}\label{eq:residual-half}
 |G\cap W|\le |W|/2\qquad(G\in\mathcal G).
\end{equation}
Initialize $\mathcal H=\mathcal G$ and $C=\varnothing$. Include all
listed domains in $\mathcal H$, including any recorded in the covering
phase: that phase removed inputs, not protocols. While $\mathcal H$
is nonempty, choose $y\in W$ belonging to at most half its members,
add $y$ to $C$, and retain only domains containing $y$.
Such a column exists, because
\[
 \sum_{y\in W}|\{G\in\mathcal H:y\in G\}|
 =\sum_{G\in\mathcal H}|G\cap W|
 \le |W|\,|\mathcal H|/2.
\]
A scan of $W$ finds it. Previously selected columns lie in every current
member of $\mathcal H$, so no column is selected twice while that list
is nonempty.

After at most $\floor{\log|\mathcal G|}+1\le a(r,s)$ iterations,
no listed domain contains all of $C$. The initial list here was
nonempty, so $C$ is nonempty. Its depth exceeds $s$ by
\eqref{eq:list-containment}. Lift the chosen representatives to their
original indices. All remaining work is polynomial in the table and
explicit list sizes, which proves the time bound.
\end{proof}

The algorithm never separately computes the optimal depth of its
returned restriction. Its lower bound follows from the proved coverage
of the list, not from a small independently checkable certificate.

\begin{lemma}[A constructive two-sided alternative]
\label{lem:constructive-two-sided}
Suppose $\rank(M)\le r$ with integer $r\ge1$, let $s\ge0$ be an integer, and put
$\lambda=\ceil{\log r}$. In time
\[
 2^{O(B(r)2^{s+\lambda})}T^{O(1)},
\]
a deterministic algorithm returns either a global protocol of depth
at most $s+2\lambda$, or nonempty original-index sets $U,V$ with
\[
 |U|\le B(r)2^s,\qquad |V|\le B(r)2^{s+\lambda},
 \qquad \D(M[U,V])>s.
\]
\end{lemma}

\begin{proof}
Apply Lemma~\ref{lem:constructive-one-sided} at threshold $s+\lambda$.
A protocol outcome has depth at most $s+2\lambda$. Otherwise we obtain
columns $V$, of size at most $B(r)2^{s+\lambda}$, with
$\D(M[X,V])>s+\lambda$. Apply the same algorithm to the transpose
of this actual restriction, now at threshold $s$ and with rank bound $r$.
Its protocol outcome would have depth at most $s+\lambda$, a contradiction.
It therefore returns rows $U$ with the required size and hardness.
The second table and its parameters are no larger than those of the first
call. Both calls retain original indices, proving all assertions.
\end{proof}

\subsection{Near-maximal extraction}

The final algorithm extracts a nonsingular minor or uses
\cref{lem:constructive-two-sided}, according to the rank. Exact rational
elimination finds the rank and original-index minors in polynomial
table time.

\begin{proof}[Proof of Theorem~\ref{thm:constructive-main}]
We first check two elementary reasons why the whole matrix might have
depth below $d$. Compute $r=\rank(M)$. If $r<d-1$, use
\cref{lem:rank-facts} to construct a protocol of depth at most
$r+1<d$: Alice names her row value and Bob transmits the answer.
The distinct row values are read directly from the table and there are
at most $2^r$ of them. Rank zero is constant and needs no communication.
If either input set has at most $2^{d-2}$ elements, that player can
instead name the input and the other player transmits the answer, using
at most $d-1$ bits. Return this protocol in either case.
These explicit protocol trees have at most $2^{O(d)}$ nodes, and their
message maps can be constructed from the table within the claimed runtime.

We may now assume $r\ge d-1\ge3$ and that both input sets have more
than $2^{d-2}$ elements. Write
\[
 j=\floor{\log r},\qquad h=j+1=\ceil{\log(r+1)},\qquad
 \ell=\ceil{\log(d+1)},
\]
and
\[
 q=\max\left\{\frac{\varepsilon d-4}{3},
                 \frac{\varepsilon d-\ell-3}{2}\right\}.
\]
In particular $j\ge1$.

If $j\ge q$, take $k=\min\{j,d-2\}$ and extract a nonsingular minor
of order $2^k$. It has depth exactly $k+1$. Since
$q<d/2\le d-2$, we have $1\le k\le d-2$ and $k\ge q$.
Exact rank computation and minor extraction have polynomial table cost.

Suppose instead that $j<q$. Put
\[
 \tau=\ceil{\log(h+1)},\qquad s=d-2h-\tau-2.
\]
We check that this is a permitted depth threshold. Since $j<q<d/2$,
we have $h\le d$, and consequently $\tau\le\ell$. Also $\tau\le h=j+1$.
It follows that
\[
 2j+\tau+3\le\min\{3j+4,\,2j+\ell+3\}<\varepsilon d.
\]
For the strict inequality, $j<q$ means that $j$ is smaller than at
least one of the two terms defining $q$. Thus
\[
 s+1=d-2j-\tau-3>(1-\varepsilon)d>0.
\]
As $s$ is an integer, $s\ge0$.

Apply Lemma~\ref{lem:constructive-two-sided} at threshold $s$, using
the actual rank $r$. Its announcement cost
$\lambda=\ceil{\log r}$ is at most $h$. A protocol outcome would have
depth at most
$s+2\lambda\le s+2h=d-\tau-2<d$; in this case, return that protocol.
Otherwise the lemma returns a restriction of depth at least $s+1$.
Its larger side, and the parameter controlling its running time, are
bounded by
\[
 B(r)2^{s+\lambda}
 \le (h+1)2^h2^{s+h}
 =(h+1)2^{d-\tau-2}\le2^{d-2}.
\]
Here $B(r)\le(r+1)(h+1)\le2^h(h+1)$ and $h+1\le2^\tau$.
Consequently this call takes $2^{O(2^d)}T^{O(1)}$ time.

Set $k=d-2$. The initial dimension check ensures that both original input
sets have more than $2^{d-2}$ elements. Enlarge the returned restriction
to an original square of exactly that size. Its depth is at least
$s+1>(1-\varepsilon)d>(1-\varepsilon)(k+1)$, and $k\ge q$.
All selected indices are original. The rank tests, padding, and exact
rational arithmetic for the supplied $\varepsilon$ fit the claimed
running time for fixed $\varepsilon$.
\end{proof}

\Needspace{12\baselineskip}
\section{Further questions}\label{sec:questions}

Can every hard Boolean matrix be restricted,
at a linear input-length scale, to a $2^k$-by-$2^k$ square of complexity
exactly $k+1$? A nonsingular minor already does this when the logarithm
of the rank is linear in the original complexity. Can near-maximal
condensation be found substantially faster, perhaps in time polynomial
in the original table size? That question asks
for hard original inputs, not for the exact value of $\D(M)$.
For the multiparty extension in \cref{sec:multiparty}, can one retain
a larger fraction of the communication complexity?

\paragraph{AI assistance.}\label{sec:ai-assistance}
Under the author's direction, the large language model Astra assisted
with the results, proofs, and exposition. Claude models also assisted
with review and editorial revisions. The \hyperref[sec:ai-methodologies]{AI methodologies statement}
at the end details this use.

\appendix
\crefalias{section}{appendix}
\section{Further two-party bounds}\label{app:tradeoffs}

The main text gives both nearly maximal hardness and a bound on the
retained fraction of the original complexity. We record the local-to-global
form used in the introduction and a faster constructive tradeoff, followed
by the exhaustive-search comparison used
in the introduction. None is needed for the main proofs.

\subsection{A local-to-global theorem}

Suppose every sufficiently small restriction is easy. A small rank
witness then forces the rank of the whole matrix to be small, and the
protocol cover applies.

\begin{theorem}[Local-to-global principle]\label{thm:local-global}
Let $s\ge1$ be an integer. If every nonempty submatrix of $M$
with at most $2^{4s}$ rows and at most $2^{4s}$ columns has
depth at most $s$, then
\[
 \D(M)\le3s-2.
\]
Equivalently, if $\D(M)>3s-2$, some submatrix within those
size bounds has depth at least $s+1$.
\end{theorem}

\begin{proof}
If $\rank(M)>2^{s-1}$, a nonsingular minor of order
$2^{s-1}+1$ has depth greater than $s$
by \cref{lem:rank-facts}. It lies within the assumed size
bounds, a contradiction. Thus $\rank(M)\le r_0=2^{s-1}$.
Rank zero is the all-zero matrix and already has depth zero.

Use \cref{lem:two-sided} with the rank bound $r_0$,
$H=\ceil{\log r_0}=s-1$.
For $r\ge1$ we have $\ceil{\log(r+1)}\le r$, hence
$B(r)\le r(r+1)+2\le(r+1)^2$.
Since $(r_0+1)^2\le2^{2s}$, the two size bounds are at most
\[
 a(r_0,s)\le2^{3s},
 \qquad a(r_0,2s-1)\le2^{4s-1}.
\]
Both fit within the local hypothesis. The existence of the
hard restriction in \cref{lem:two-sided} is therefore
impossible, giving
$\D(M)\le s+2(s-1)=3s-2$.
\end{proof}

The hypothesis says \emph{at most} these dimensions. If a matrix has
fewer rows or columns, it still tests its existing restrictions.
There is no vacuous requirement for a square larger than the matrix.

\subsection{An alternative constructive tradeoff}

Another choice of parameters in the same extraction algorithm trades
retained hardness for a faster running time. It again returns either a
short global protocol or a small hard restriction. This yields a condensate
at input length within five bits of a supplied hardness lower bound, retaining
one-quarter hardness, with a faster running time than the near-maximal
guarantee of \cref{thm:constructive-main}. As in
\cref{sec:construction}, $T$ denotes the number of entries in the explicit
input matrix.

\begin{theorem}[A constructive local-to-global alternative]
\label{thm:constructive-alternative}
For every integer $t\ge1$, a deterministic algorithm returns either
a valid protocol for $M$ of depth at most $3t-2$, or nonempty original
row and column sets $U,V$ with
\[
 |U|,|V|\le2^{4t},\qquad \D(M[U,V])>t.
\]
Its running time is $2^{O(t8^t)}T^{O(1)}$. The two alternatives need
not be disjoint.
\end{theorem}

\begin{proof}
Compute the real rank by exact rational linear algebra. If it is zero,
return the all-zero protocol. Put $r_0=2^{t-1}$. If the rank exceeds
$r_0$, find $r_0+1$ independent original columns and then that many
independent original rows of their restriction. The resulting nonsingular
minor has depth greater than $t$ by Lemma~\ref{lem:rank-facts}, and its
dimensions are at most $2^{4t}$. Exact elimination has polynomial bit
complexity on a Boolean table; its pivot indices select original inputs,
not linear combinations of them.

Otherwise apply Lemma~\ref{lem:constructive-two-sided} with $r=r_0$
and $\lambda=t-1$. Its protocol has depth at most $3t-2$. For its
restriction outcome, the inequalities
\[
 B(r_0)\le(r_0+1)^2\le2^{2t}
\]
give row and column bounds at most $2^{3t}$ and $2^{4t-1}$.
For the running time, use the sharper estimate
\[
 B(r_0)\le(r_0+1)(t+1)\le(t+1)2^t.
\]
The exponent in the two-sided algorithm is therefore at most a constant
times $(t+1)2^t2^{2t-1}=O(t8^t)$. Rank computation and original-index
bookkeeping have polynomial table cost.
\end{proof}

\begin{corollary}[Condensation at almost full input length]
\label{cor:constructive-full-scale}
Given $M$ and an integer $d\ge6$ with $\D(M)\ge d$, put
$k=4\floor{(d-2)/4}$. A deterministic algorithm returns an original
$2^k$-by-$2^k$ submatrix $N$ such that
\[
 d-5\le k\le d-2,\qquad \D(N)\ge k/4+1,
\]
in time $2^{O(d2^{3d/4})}T^{O(1)}$.
\end{corollary}

\begin{proof}
Set $t=\floor{(d-2)/4}\ge1$. Since $d>3t-2$, the algorithm of
Theorem~\ref{thm:constructive-alternative} cannot return its protocol
outcome. Its restriction has dimensions at most $2^{4t}=2^k$ and
integer depth at least $t+1$. Lemma~\ref{lem:rank-facts} gives more than
$2^{d-2}\ge2^k$ original inputs on each side. Add unused original
rows and columns to obtain an exact square of side $2^k$; enlargement
cannot reduce depth. Finally, $t8^t=O(d2^{3d/4})$.
\end{proof}

\subsection{Exhaustive search}\label{sec:exhaustive-search}

Fix a rational $0<\varepsilon<1$. Here is the direct-search comparison
behind the algorithmic statement.
First, a supplied lower bound $d\le\D(M)$, with $d\ge4$, suffices for
the existence guarantee at scale $d$. Retain an inclusion-minimal nonempty
row set $U$ with $\D(M[U,Y])\ge d$. It has at least two rows. Deleting
any one row leaves depth below $d$, and adding it back gives a protocol
of depth at most
\[
 1+\max\{\D(M[U\setminus\{x\},Y]),1\}\le d.
\]
Alice first indicates whether she holds the added row; on that branch,
Bob sends the answer. Thus $\D(M[U,Y])=d$. This is an existence argument,
not an efficient procedure for finding $U$.

Apply \cref{thm:near-maximal} to this restriction. It gives a hard square
with at most $L=2^{d-1}$ rows and columns, and the corresponding lower
bound on its input length. A direct algorithm can therefore enumerate
all original squares of power-of-two side at most $L$ and test their
depths. If $M$ has $m$ rows and $n$ columns, then $T=mn$, and there are
at most $L T^L$ such candidates: for a side $q$, the number is at most
$m^q n^q=T^q$.

The depth of a candidate can be computed exactly by dynamic programming
over its nonempty row and column subsets. A monochromatic restriction
has depth zero. Otherwise minimize one plus the larger child depth over
every nontrivial bipartition of its rows or of its columns. These are
exactly the possible first messages of a deterministic protocol.
Smaller restrictions are processed first. For a candidate with at most
$L$ rows and columns, the states and all tested bipartitions together
number $2^{O(L)}$. The complete direct search consequently takes
$2^{O(2^d)}T^{O(2^d)}$ time.

\paragraph{Rank preprocessing.}
An elementary preprocessing step removes the parameter-dependent exponent
of $T$, although it gives a weaker bound in $d$ than
\cref{thm:constructive-main}. Compute $r=\rank(M)$ and put
\[
 j=\floor{\log r},\qquad
 q=\max\left\{\frac{\varepsilon d-4}{3},
 \frac{\varepsilon d-\ceil{\log(d+1)}-3}{2}\right\}.
\]
By \cref{lem:rank-facts}, $r\ge d-1\ge3$, so $j\ge1$; also
$q<\varepsilon d/2<d/2\le d-2$. If $j\ge q$, choose
$k=\min\{j,d-2\}$ and extract a nonsingular minor of order $2^k$.
Its depth is exactly $k+1$, and it satisfies all the bounds in
\cref{thm:constructive-main}. Exact Gaussian elimination and minor
extraction have polynomial bit cost in $T$.

Otherwise $j<q$, and
\[
 r<2^{j+1}<2^{q+1}<2^{\varepsilon d/2+1}.
\]
Keep one original representative of each distinct row and column value.
This preserves rank and communication complexity and leaves at most
$2^r$ rows and $2^r$ columns, by \cref{lem:rank-facts}.
To see that the smaller table contains a suitable square with $k=d-2$,
apply the depth-$d$ row-restriction existence argument above to this
table. If the resulting depth-$d$ row restriction has rank $r'$, then
$\floor{\log r'}\le j<q$.
The proof of \cref{thm:near-maximal} therefore takes its low-rank
case and gives a square with $k=d-2$ and depth greater than
$(1-\varepsilon)d$. Since $k\ge q$, this square meets all the bounds
in \cref{thm:constructive-main}. The argument guarantees its existence;
the algorithm need not find the intermediate row restriction.

Enumerate all squares of side $L=2^{d-2}$ in the table of representatives,
testing their depths by the dynamic program above. There are at most
$2^{2rL}$ candidates, each test costs $2^{O(L)}$, and the preprocessing
takes polynomial time in $T$. Return a candidate meeting the desired
depth bound, using its retained original indices. The total running time is
\[
 2^{O(r2^d)}T^{O(1)}
 \le 2^{O(2^{(1+\varepsilon/2)d})}T^{O(1)}.
\]
Thus rank preprocessing already gives a constant exponent on table size.
The algorithm in \cref{thm:constructive-main} improves the dependence on
$d$ to $2^{O(2^d)}$ and makes the protocol-domain construction executable.
This comparison concerns the square-extraction guarantee under a supplied
promise $\D(M)\ge d$. The exhaustive search uses that promise; the algorithm
in \cref{thm:constructive-main} also works without it, with the additional
possibility of returning a global protocol of depth below $d$. Its square
outcome does not by itself certify the promise. This search comparison
uses the existence theorem, rather
than providing an independent proof of condensation.

\section{Several players and finite output alphabets}
\label{sec:multiparty}

The two-party proof had one decisive communication step: Bob named a
set containing his input, and the players ran the short protocol valid
on that set. With more players, each input holder can make the same
announcement in turn. Membership still has to be recognizable from
that player's input alone. We show that this retains at least a
$1/(p+1)$ fraction of the communication complexity for any fixed number
$p$ of players, even when the required output is taken from an
arbitrary finite alphabet.

We use the number-in-hand model with public communication, as in
Draisma, Kushilevitz, and Weinreb~\cite[Section~2]{DraismaKushilevitzWeinreb11};
we specify the finite-output convention explicitly.
Let $F:X_1\times\cdots\times X_p\to Z$ be a total function, where all
sets are finite and nonempty. Player $i$ receives $x_i$. Communication is
written on a public blackboard; the next speaker is determined by the
transcript, and a transmitted bit depends only on that speaker's input
and the transcript. Leaves have fixed labels in $Z$. Write $D_p(F)$ for
the minimum worst-case number of communicated bits. A restriction selects
an original subset of each $X_i$. In particular, it neither combines
players nor changes the required output.

\Needspace{10\baselineskip}
\begin{theorem}[Multiparty condensation]\label{thm:multiparty-main}
Let $p\ge2$ and let $F:X_1\times\cdots\times X_p\to Z$ be a total
function on finite nonempty sets. Write $c=D_p(F)$ and suppose
$c\ge p+2$. Put $s=\floor{(c-1)/(p+1)}$. There are nonempty original subsets
$X'_i\subseteq X_i$ such that
\[
 |X'_i|\le 2^{(p+2)s+1}\quad(1\le i\le p),\qquad
 D_p(F|_{X'_1\times\cdots\times X'_p})\ge\ceil{c/(p+1)}.
\]
\end{theorem}

We prove this through the following local-to-global statement.

\Needspace{10\baselineskip}
\begin{theorem}[Multiparty local-to-global principle]
\label{thm:multiparty-condensation}
Let $p\ge2$ and $s\ge1$ be integers. If every product restriction
$F|_{X'_1\times\cdots\times X'_p}$ with
\[
 |X'_i|\le 2^{(p+2)s+1}\qquad(1\le i\le p)
\]
has communication complexity at most $s$, then
\[
 D_p(F)\le(p+1)s.
\]
\end{theorem}

For fixed $p$, the restricted inputs in \cref{thm:multiparty-main} need only $O_p(c)$ bits
to index, and their communication complexity remains $\Omega_p(c)$.
There is no restriction on the size of the finite output alphabet.
The proof needs two additions to the two-party argument. We must count
leaf conditions when there are many possible outputs, and we must
extend one player's input set without giving that player access to the
others' inputs. Output indicators solve the first problem; successive
applications of the covering lemma solve the second.

\Needspace{8\baselineskip}
\subsection{Output indicators}

Fix player $i$ and write $x_{-i}$ for the tuple of all inputs except
$x_i$. Hold those other inputs at a tuple $u$.
The dependence of the output on $x_i$ can be recorded by asking, for each
$z\in Z$, whether the output equals $z$. For example, if inputs $a,b$
give outputs red and blue at this tuple, the corresponding part of the
indicator table is
\[
\begin{array}{c|cc}
 &a&b\\ \hline
 (u,\text{red})&1&0\\
 (u,\text{blue})&0&1
\end{array}
\]
The two rows express the two possible correctness tests at a leaf:
is the required output red, or is it blue? They are questions about the
\emph{same} original function. No player receives the extra output
label as input, and the protocol does not transmit it separately.

Formally, define the Boolean matrix
\begin{equation}
 H_i((x_{-i},z),x_i)=[F(x_1,\ldots,x_p)=z].
 \label{eq:one-hot-flattening}
\end{equation}
Here brackets denote the indicator of the enclosed statement. This is
the one-hot flattening of $F$ for player $i$. It is used only in the
analysis. Two values of $x_i$ give identical columns of $H_i$ exactly
when they give the same output of $F$ for every choice of the other
inputs. We may keep one representative of each such column and restore
all copies by a local map of player $i$.

Suppose $\rank(H_i)\le r$, where $r\ge1$. By
Lemma~\ref{lem:rank-facts}, there are at most $2^r$ different columns,
and at most $2^r$ different rows. For completeness, the column bound
comes from selecting $r$ actual rows spanning the row space: their
Boolean values determine the entire column. The analogous argument
gives the row bound.

\subsection{Compatible input sets}

To leave only player $i$'s choices open, fix a complete binary
depth-$s$ protocol tree, its speakers and outputs, and every other
player's local maps. An input $x_i$ is \emph{compatible} with this
partial protocol if one assignment of bits at player $i$'s nodes makes
the protocol correct for every retained tuple of the other inputs.
The assignment can depend on $x_i$, but not on any of those other
inputs. This is precisely the common-assignment requirement in
Lemma~\ref{lem:compatibility}.

Every compatible input set supports a genuine $p$-player protocol of
depth $s$: choose one successful assignment for each input $x_i$ and use
its bit at each node as player $i$'s local map. All other players keep
their fixed, individually local maps. Conversely, a correct protocol
on any set of player $i$'s inputs supplies such a partial protocol whose
compatible set contains them.

\begin{lemma}[Counting compatible inputs]
\label{lem:multiparty-count}
Let $r\ge1$ and suppose $\rank(H_i)\le r$.
Use $B(r)=(r+1)\ceil{\log(r+1)}+2$ as in \eqref{eq:counting-coefficient}.
The compatible input sets of all depth-$s$ partial protocols with
player $i$'s maps left open number at most
\[
 2^{B(r)2^s-1}.
\]
The count includes arbitrary retained sets of the other inputs and
all output labels in $Z$.
\end{lemma}

\begin{proof}
At a leaf labelled $z$, compatibility with a retained set $U$ of
other-input tuples means
\[
 H_i((u,z),x_i)=1\qquad\text{for every }u\in U.
\]
Choose an inclusion-minimal collection of these row constraints giving
the same condition. If it has size $q$, each constraint has a witness
column satisfying all the others but not that constraint. For $q\ge2$
the witnesses form $J_q-I_q$, whose real rank is $q$. Hence $q\le r$;
the cases $q=0,1$ also satisfy this bound.
Now use the row-sum encoding of \cref{lem:leaf-count}, applied to $H_i$.
The sum of the selected rows is determined by at most $r$ basis-column
entries, each in $\{0,\ldots,r\}$. A column satisfies the leaf condition
exactly when this sum equals $q$. Including the choice of $q$ gives at most
\[
 (r+1)^{r+1}
\]
different leaf conditions. Crucially, this already counts all output labels:
$(u,z)$ is the index of one of the rows just counted, not an extra
choice outside the count.

The same recursion as in Lemma~\ref{lem:domain-count} now applies.
Track the other-input tuples consistent with their owners' fixed maps
on the path. A node of player $i$ takes the OR of the two child
conditions, since one child must work for all these tuples. A node of
another player partitions the tuples according to that player's fixed
bit and takes AND. Successful assignments below different children
use disjoint nodes of player $i$, so they combine legally. There are
$2^s-1$ internal gates and $2^s$ leaves. Counting the two possible gates
and the preceding leaf conditions gives the stated bound, with slack
in $B(r)$.
\end{proof}

The point of the count is the same as before: if every small selection
of inputs is compatible with some short protocol, then a small number
of those protocols cover the entire input set. The announcing player
can choose a covering set using only its own input.

\begin{lemma}[Extension of one input set]
\label{lem:multiparty-one-sided}
Let $r\ge1$ and suppose $\rank(H_i)\le r$. If every restriction of
$X_i$ to at most $B(r)2^s$ values, with the other input sets unchanged, has a
depth-$s$ protocol, then
\[
 D_p(F)\le s+\ceil{\log r}.
\]
\end{lemma}

\begin{proof}
Retain one representative of each of the at most $2^r$ column values
of $H_i$. Set $K=B(r)2^s$. The local hypothesis says that every subset
of at most $K$ representatives lies in a compatible set. There are at
most $2^{K-1}$ such sets by Lemma~\ref{lem:multiparty-count}.
The finite covering argument of Lemma~\ref{lem:cover} gives a cover
by at most $r$ compatible sets. Indeed, if no set covered more than
half a nonempty remainder, $K$ independent uniform samples from that
remainder would all lie in some compatible set with probability at
most $2^{K-1}2^{-K}=1/2$, contradicting the local hypothesis. After
$r$ strict halvings fewer than one of the original $2^r$
representatives remain.

Player $i$ names the first covering set containing its representative,
using $\ceil{\log r}$ bits. The players then run that set's protocol.
Both the representative map and the covering-set choice depend only
on $x_i$. Thus this announcement does not give any player access to
another player's input.
\end{proof}

For three players, for example, first hold the second and third input
sets restricted and extend the first. This works for every permitted
choice of those two restrictions. We can therefore leave the first
player unrestricted while extending the second, and finally extend
the third. The next lemma records the caps needed to pay for each
preceding announcement. No single protocol has to work for all choices
of the still-restricted sets; each choice may use its own protocol.

\Needspace{13\baselineskip}
\begin{lemma}[Extension of all input sets]
\label{lem:multiparty-lifting}
For each $i$, let $r_i\ge1$ and suppose $\rank(H_i)\le r_i$. Write
$\lambda_i=\ceil{\log r_i}$ for player $i$'s announcement cost.
Fix an ordering of the players and set
\[
 K_i=B(r_i)2^{s+\sum_{j<i}\lambda_j}.
\]
If every product restriction with $|X'_i|\le K_i$ has depth at most
$s$, then
\[
 D_p(F)\le s+\sum_{i=1}^p\lambda_i.
\]
\end{lemma}

\begin{proof}
Keep arbitrary restrictions of players $2,\ldots,p$ within their
caps. Lemma~\ref{lem:multiparty-one-sided} extends player $1$'s input
set to all of $X_1$, at cost $\lambda_1$. This holds for every choice of
the still-restricted sets. Therefore, with player $1$ unrestricted,
every allowed restriction of player $2$ has depth at most $s+\lambda_1$.
Its cap is exactly the one required to apply the lemma at this new
depth. Extend player $2$, and continue in order. Restricting inputs
cannot increase the ranks of the full matrices $H_i$, so the same
rank bounds remain valid throughout. Each announcement is made by
the owner of the input being extended.
\end{proof}

\Needspace{8\baselineskip}
\subsection{Rank witnesses}

It remains to remove the rank assumptions. As in the two-party proof,
large rank would already be visible on a small restriction. Here a row
of $H_i$ names several inputs and an output, so we must check that a
matrix minor really comes from small original input sets for every player.
The rank argument is the standard one: a transcript determines a product
of local input sets, which contributes a rank-one matrix after grouping
the inputs into two sides; see~\cite[Section~3]{DraismaKushilevitzWeinreb11}.
Here we apply it to the output-indicator matrices $H_i$.

\begin{lemma}[Flattening rank on a small restriction]
\label{lem:multiparty-rank}
If every product restriction with at most $2^s+1$ inputs per player
has depth at most $s$, then $\rank(H_i)\le2^s$ for every $i$.
\end{lemma}

\begin{proof}
A leaf of a blackboard protocol corresponds to a product of input
sets and has a fixed output label. Its indicator in $H_i$ factors
into the indicator of player $i$'s set and the indicator of the
other players' product together with that output label. It therefore
has rank at most one. Summing over at most $2^s$ leaves gives
$\rank(H_i)\le2^s$ for any function with communication complexity at most $s$.

If a full $H_i$ had larger rank, take a nonsingular minor of order
$2^s+1$. Its columns use at most that many inputs of player $i$.
Its rows use at most that many tuples of other inputs. Project those
tuples to each of the other input sets and take the product of the
selected subsets. Every player has at most $2^s+1$ selected inputs,
and the new one-hot flattening still contains the minor. The output
tags remain row indices; they require no additional input values.
This contradicts the rank bound for the hypothesized short protocol
on that restriction.
\end{proof}

\begin{proof}[Proof of Theorem~\ref{thm:multiparty-condensation}]
The hypothesis includes all restrictions needed by
Lemma~\ref{lem:multiparty-rank}. Thus take $r_i=r=2^s$ and
$\lambda_i=s$ in Lemma~\ref{lem:multiparty-lifting}. Since $s\ge1$,
\[
 B(r)\le r^2+r+2\le2r^2=2^{2s+1}.
\]
The cap for player $i$ is at most
\[
 B(r)2^{s+(i-1)s}\le2^{(i+2)s+1}\le2^{(p+2)s+1}.
\]
The extension lemma therefore gives $D_p(F)\le(p+1)s$.
\end{proof}

\begin{proof}[Proof of Theorem~\ref{thm:multiparty-main}]
Since $c\ge p+2$, the integer $s=\floor{(c-1)/(p+1)}$ is at least one
and satisfies
\[
 (p+1)s=(p+1)\floor{(c-1)/(p+1)}\le c-1<D_p(F).
\]
By Theorem~\ref{thm:multiparty-condensation}, the local hypothesis must
fail: some product restriction with $|X'_i|\le2^{(p+2)s+1}$ for every
$i$ has communication complexity at least $s+1=\ceil{c/(p+1)}$.
\end{proof}

For Boolean outputs the same argument has slightly better constants.
They also recover the two-party local-to-global bound.

\begin{proposition}[Boolean outputs]
\label{prop:multiparty-boolean}
For integers $p\ge2,s\ge1$ and a total Boolean function $F$, if
every product restriction with at most $2^{(p+2)s}$ inputs per player
has depth at most $s$, then
\[
 D_p(F)\le(p+1)s-p.
\]
\end{proposition}

\begin{proof}
Use the ordinary Boolean flattening
$M_i(x_{-i},x_i)=F(x)$, without output tags. After padding a
depth-$s$ protocol to full depth, the inputs yielding output $1$
below each parent of two leaves form either the empty set or a product
of input sets. Each such set has flattening rank at most one, so
$\rank(M_i)\le2^{s-1}$. A larger-rank minor would again appear in
a product restriction of at most its order on each side. The local
hypothesis therefore bounds every full flattening rank by
$r=2^{s-1}$.

Lemma~\ref{lem:domain-count} bounds both zero and one leaf conditions
directly for these Boolean flattenings. The preceding extension
argument applies with announcement cost $\lambda=\ceil{\log r}=s-1$ and
$B(r)\le r(r+1)+2\le2^{2s}$. Every required cap is at most
\[
 B(r)2^{s+(i-1)(s-1)}\le2^{(p+2)s}.
\]
The resulting depth is $s+p(s-1)$, as claimed.
\end{proof}

This existential extension concerns a fixed number of number-in-hand
players, with constants depending on that number. It restricts each
player's input set while retaining all players. The constructive
algorithm proved in the main text concerns two parties.
\Cref{sec:questions} asks whether a larger fraction can be retained.

\section*{Acknowledgments and provenance}\label{sec:provenance}

This investigation began with the problem, raised by Yao~\cite{Yao79}
and studied by Kushilevitz and Weinreb~\cite{KushilevitzWeinreb09}, of
determining deterministic communication complexity from a Boolean truth
table. Joint work with Serge Gaspers and Tao Zixu He established
NP-hardness of exact computation~\cite{GaspersHeMackenzie25}, concurrently
with the independent proof of Hirahara, Ilango, and
Loff~\cite{HiraharaIlangoLoff25}. Further
joint work with Gaspers and He, in preparation, obtains hardness of
additive approximation under the Exponential Time Hypothesis. Both works
with Gaspers and He use interlacing, a matrix construction that
amplifies communication gaps and originates in joint work with Abdallah
Saffidine on direct sums~\cite{MackenzieSaffidine25}; the exact-hardness
paper develops a relaxed form of it with polynomial-size column sets.

The next aim was NP-hardness of additive approximation under randomized
reductions. This required a refined, precisely defined relaxation of
interlacing that could be sparsified to keep the reductions polynomial
in size for every fixed additive gap. A joint manuscript with Gaspers
and He on these randomized reductions uses protocol domains, with a
rank-sensitive count adapted from known methods, to sample interlaced
matrices. Studying the domains on which complete protocols remain valid
led to the condensation results here. That manuscript is unpublished;
this paper states the definitions and proves the counts it needs, and
uses none of that manuscript's results.

\section*{AI methodologies}\label{sec:ai-methodologies}

The large language model Astra was used in developing and revising the mathematical results and
proofs in \cref{sec:domains,sec:counting,sec:condensation,sec:construction}
and \cref{app:tradeoffs,sec:multiparty},
including the counting arguments, the two-party condensation theorems,
the extraction algorithm and its analysis, and the finite-output
multiparty extension. It was also used in drafting and revising the
exposition throughout the paper, including the abstract and introduction.
Anthropic's Claude models were also used to review proof drafts and
exposition and to suggest editorial revisions.

The author takes responsibility for the claims,
proofs, citations, and presentation.

\clearpage
\bibliographystyle{alphaurl}
\begingroup
\interlinepenalty=10000
\bibliography{references}
\endgroup
\end{document}